\documentclass[aps,twocolumn,superscriptaddress,nofootinbib,a4paper,longbibliography]{revtex4-2}

\usepackage{times}
\usepackage{epsfig}
\usepackage{amsfonts}
\usepackage{amsmath}
\usepackage{amsthm}
\usepackage{amssymb}
\usepackage{amsthm}
\usepackage{dsfont}
\usepackage{bm}
\usepackage{enumerate}
\usepackage{mathtools}
\usepackage{color}
\usepackage{multirow}
\usepackage[normalem]{ulem}
\newcommand{\stkout}[1]{\ifmmode\text{\sout{\ensuremath{#1}}}\else\sout{#1}\fi}
\usepackage{latexsym}
\usepackage{mathrsfs}
\usepackage{natbib}
\usepackage{verbatim}
\usepackage[T1]{fontenc}
\usepackage{float}
\usepackage{graphicx}
\usepackage{subcaption}
\usepackage{xcolor}
\usepackage{physics}
\usepackage{soul}
\usepackage{caption}
\usepackage{subcaption}

\usepackage[font=small,labelfont=bf]{caption}

\usepackage{xspace}

\newtheorem{theorem}{Theorem}

\newtheorem{lemma}[theorem]{Lemma}

\renewcommand{\ket}[1]{|#1\rangle}
\renewcommand{\bra}[1]{\langle#1|}
\newcommand{\bracket}[3]{\langle#1|#2|#3\rangle}
\newcommand{\ketbr}[1]{\ket{#1}\!\bra{#1}}

\makeatletter
\newcommand{\pushright}[1]{\ifmeasuring@#1\else\omit\hfill$\displaystyle#1$\fi\ignorespaces}
\newcommand{\pushleft}[1]{\ifmeasuring@#1\else\omit$\displaystyle#1$\hfill\fi\ignorespaces}
\makeatother

\usepackage{amssymb}
\usepackage{amsthm}
\usepackage{mathtools} 
\usepackage[customcolors]{hf-tikz} 
\usetikzlibrary{patterns}
\usetikzlibrary{matrix,decorations.pathreplacing}

\pgfkeys{tikz/mymatrixenv/.style={decoration={brace},every left delimiter/.style={xshift=8pt},every right delimiter/.style={xshift=-8pt}}}
\pgfkeys{tikz/mymatrix/.style={matrix of math nodes,nodes in empty cells,left delimiter={[},right delimiter={]},inner sep=1pt,outer sep=1.5pt,column sep=2pt,row sep=2pt,nodes={minimum width=20pt,minimum height=10pt,anchor=center,inner sep=0pt,outer sep=0pt}}}
\pgfkeys{tikz/mymatrixbrace/.style={decorate,thick}}

\usepackage[colorlinks=true,linkcolor=blue,citecolor=magenta,urlcolor=blue]{hyperref}

\begin{document}
	

	\title{Maximally non-projective measurements are not always symmetric informationally complete }

	\author{Gabriele Cobucci}\thanks{These authors contributed equally.}
	\affiliation{Physics Department and NanoLund, Lund University, Box 118, 22100 Lund, Sweden.}
	
	\author{Raphael Brinster}\thanks{These authors contributed equally.}
	\affiliation{Institute for Theoretical Physics III, Heinrich Heine University D\"usseldorf, D-40225 D\"usseldorf, Germany}
	
	\author{Shishir Khandelwal}
	\affiliation{Physics Department and NanoLund, Lund University, Box 118, 22100 Lund, Sweden.}

	\author{Hermann Kampermann}
	\affiliation{Institute for Theoretical Physics III, Heinrich Heine University D\"usseldorf, D-40225 D\"usseldorf, Germany}

	\author{Dagmar Bru\ss }
	\affiliation{Institute for Theoretical Physics III, Heinrich Heine University D\"usseldorf, D-40225 D\"usseldorf, Germany}
	
	\author{Nikolai Wyderka}
	\affiliation{Institute for Theoretical Physics III, Heinrich Heine University D\"usseldorf, D-40225 D\"usseldorf, Germany}
	
	\author{Armin Tavakoli}
	\affiliation{Physics Department and NanoLund, Lund University, Box 118, 22100 Lund, Sweden.}
	
	\begin{abstract}
		Standard quantum measurements are projective. However, the full scope of quantum measurements is represented by positive operator-valued measures (POVMs) and many of these break the limitations of projective measurements as resources in quantum information. It is therefore natural to consider how accurately an experimenter with access only to projective measurements and classical processing can simulate POVMs. The most well-known class of non-projective measurements  is called symmetric informationally complete (SIC). Such measurements are both ubiquitous in the broader scope of quantum information theory and known to be the most strongly non-projective measurements in qubit systems. Here, we show that beyond qubit systems, the SIC property is in general not associated with the most non-projective measurement. For this, we put forward a semidefinite programming criterion for detecting genuinely non-projective measurements. This method allows us to determine quantitative simulability thresholds for generic POVMs and to put forward a conjecture on which qutrit and ququart measurements that are most strongly non-projective.

	\end{abstract}

	\date{\today}

	\maketitle

	\textit{Introduction.---}
	Standard measurements in quantum theory are complete sets of orthogonal projectors. However, it is well-established that the full scope of quantum measurements is represented by the more general notion of positive operator-valued measures (POVMs). On a given Hilbert space, there exist non-projective POVMs that cannot be reduced to projective measurements \cite{DAriano2005}. The seminal example of such a measurement is the symmetric informationally complete (SIC) POVM \cite{Zauner1999, Renes2004}. A $d$-dimensional SIC-POVM has the largest possible  number of irreducible outcomes and all the angles between its measurement operators are identical. These measurements are broadly important in quantum theory. For instance, they are optimal for both state tomography \cite{Scott2006}  and randomness generation \cite{Acin2016}, they play a key role in some approaches to the foundations of quantum theory \cite{Fuchs2013} and they serve as tools for  entanglement witnessing \cite{Shang2018, Tavakoli2024a}, contextuality \cite{Bengtsson2012}, nonlocality \cite{Gisin2019, Huang2021}, self-testing \cite{Tavakoli2019, Mironowicz2019} etc. A plethora of experimental platforms developed for implementing non-projective measurements have primarily focused on realising SIC-POVMs \cite{Zhao2015, Bian2015, Hou2018, Smania2020, Stricker2022, Wang2023, Feng2025}. Their existence in all dimensions is a famous open-ended conjecture \cite{Zauner1999, Fuchs2017, Horodecki2022}.
	
	In view of this, it may be expected that SIC-POVMs are, in the appropriate sense, the POVMs that most strongly defy the limitations of standard projective measurements. This was conjectured in Ref.~\cite{Oszmaniec2017}. The same work also formalised the study of simulating POVMs using only projective measurements and classical resources. This is not only  conceptually motivated  but it is also practically relevant since non-projective measurements are considerably more expensive to implement than projective measurements; they need ancilla systems and entangling gates. As in most resource theories, there are infinitely many mathematically legitimate ways of quantifying projective simulability. However, for projective simulability the standard benchmark employed in the literature  is based on a depolarising channel (see e.g.~\cite{Oszmaniec2017, Hirsch2017, Guerini2017, Singal2022, Martinez2023,kotowski2025prettygoodsimulationquantummeasurements}). This is because the depolarisation is unbiased with respect to the POVM under consideration, experimentally meaningful and mathematically elegant. The extent to which a given POVM is genuinely non-projective is then the smallest amount of depolarisation noise that it can be exposed to before admitting a projective simulation.

	For qubit systems, it has been proven that SIC-POVMs indeed  are the most non-projective measurements \cite{Hirsch2017}.  Here, however, we show that the same is not in general true for dimensions higher than two. For this, we put forward a necessary condition for projective simulability which admits a semidefinite programming (SDP) formulation, a standard optimisation tool used in quantum information \cite{Tavakoli2024b}. We first deploy it to investigate the projective simulability of SIC-POVMs in dimension three and show that the SIC property is insufficient to determine the simulability thresholds, i.e.~some SIC-POVMs are more non-projective than others. Next, we consider SIC-POVMs in dimension four and show that none of them correspond to the most non-projective measurement. Remarkably, we find that all four-dimensional SIC-POVMs are less robust than their lower-dimensional counterparts. Perhaps counterintuitively, this shows that SIC-POVMs do not always become harder to simulate as we increase the dimension. 
	
	In view of these results, it is natural to ask the following question: if not SIC-POVMs, then which measurements are the most non-projective? Toward answering this, we use our SDP criterion as an oracle in a systematic numerical search for the most strongly non-projective POVM. For qutrits, we find that it returns one single SIC-POVM (among infinitely many possible). For ququarts, it returns a ten-outcome POVM which is not a SIC-POVM but bears a clear conceptual resemblance to it. This leads us to investigate this new type of ``flagged'' measurement and its capability of defying projective simulations.

	\textit{Projective simulation of POVMs.---} Consider that we are given an $n$-outcome POVM $\mathbf{E}=\{E_a\}_{a=1}^n$ on a $d$-dimensional Hilbert space, where $a$ denotes the outcome. Our goal is to simulate it by using only the following restricted class of resources. (i) Classical randomness, represented by a stochastic variable $\lambda$ subject to a probability distribution $\{q_\lambda\}_\lambda$. (ii) A set of projective measurements, $\{P_{k|\lambda}\}$, where $k$ represents the outcome. The variable $\lambda$ is interpreted as the measurement choice. (iii) Classical post-processing of  $k$ and $\lambda$, through a probabilistic rule $p(a|k,\lambda)$, which produces the final output $a$. If these resources can reproduce the POVM, it takes the form $
	E_a=\sum_\lambda q_\lambda \sum_{k} p(a|k,\lambda)P_{k|\lambda}$. If such a model exists, we say that $\mathbf{E}$ admits a projective simulation \cite{Oszmaniec2017}. 
	However, it turns out that one can without loss of generality omit the post-processing, i.e.~we can always select it as $p(a|k,\lambda)=\delta_{a,k}$. This leads to the simpler form
	\begin{equation}\label{PVMsim}
		E_a=\sum_\lambda q_\lambda P_{a|\lambda},
	\end{equation}
	where the projective measurements $P_{a|\lambda}$ in the simulation can also be null matrices. In the Supplemental Material \cite{supplemental_material} (see also Refs. \cite{Skrzypczyk_2023,Appleby_2005,Scott_2010,Webb_2016}), we provide a constructive proof in which it is shown how to explicitly eliminate the post-processing. Notably, such simulable POVMs can be realized without access to additional degrees of freedom, even though they are not necessarily projective. This shows that other (non-convex) measures of non-projectiveness, like unsharpness, are insufficient to capture this feature \cite{liu2021quantifying}.

	We denote the set of projectively simulable POVMs by $\mathcal{P}$. Assume that $\mathbf{E}$ is not projectively simulable ($\mathbf{E}\notin \mathcal{P}$). Then, a natural way to quantify its magnitude of non-projectivity is through applying the depolarisation channel, $\Phi_v(X)=vX+\frac{1-v}{d}\tr\left(X\right)\openone$, to each of the POVM elements\footnote{Since the depolarisation channel is self-dual, this is equivalent to performing the noise-free measurement on depolarised input states.} and determine the largest visibility, $v\in[0,1]$, for which a projective simulation is possible. The projective simulability threshold (sometimes called critical visibility) for $\mathbf{E}$ becomes
	\begin{equation}\label{eq:vstar}
		v^*(\mathbf{E})=\max\{ v| \{\Phi_v(E_a)\}_a\in \mathcal{P}\}.
	\end{equation}
	This measure can be interpreted as white noise tolerance of genuine non-projective features.

	\textit{Necessary condition for projective simulability.---} How can we determine if a given POVM $\mathbf{E}$ admits a projective simulation? A complete solution is not known. We identify a necessary condition in the form of an SDP. Such programs are known to be efficiently computable \cite{Tavakoli2024b}. The main idea is to partition the space of projective measurements over $d$-dimensional Hilbert space by the ranks of the measurement operators. Specifically, consider a vector $\vec{r}=(r_1,\ldots,r_n)$ such that $r_a$ are non-negative integers and $\sum_{a=1}^n r_a=d$. Every such vector can be associated with a distinct class of projective measurements, $\{P_a\}_a$, in which $\text{rank}(P_a)=r_a$. The union of all these classes is the set of all projective measurements. Therefore, we  call $\vec{r}$ a rank-vector and use it as an index for partitioning the space of projective measurement. Hence, we can substitute the projective measurement index $\lambda$ in Eq~\eqref{PVMsim} with the pair $(\vec{r},\mu)$, where $\vec{r}$ specifies the rank-vector and $\mu$ indexes the projective measurement within that class. Inserted into \eqref{PVMsim}, this leads to $E_a=\sum_{\vec{r}} F_{a|\vec{r}}$, where $F_{a|\vec{r}}=\sum_{\mu} q_{\vec{r},\mu} P_{a|\vec{r},\mu}$. By construction, we have that $F_{a|\vec{r}}\succeq 0$. Furthermore, using conditional probability we can write $q_{\vec{r},\mu}=q_{\vec{r}}\ q_{\mu|\vec{r}}$. Then, it follows from measurement completeness and normalisation that $\sum_a F_{a|\vec{r}}=q_{\vec{r}}\openone$. Finally, since by definition $\text{rank}(P_{a|\vec{r},\mu})=r_a$ and $P_a$ is projective, it follows that $\tr\left(P_{a|\vec{r},\mu}\right)=r_a$. Hence, $\tr(F_{a|\vec{r}})=q_{\vec{r}}\ r_a$. Putting this together, our necessary condition for $\mathbf{E}\in \mathcal{P}$ becomes
	\begin{equation}\label{SDPcond}
		\begin{aligned}
			\text{find} & \quad \{F_{a|\vec{r}}\}_{a,\vec{r}}\\
			\text{s.t.}& \quad	\sum_{\vec{r}} F_{a|\vec{r}}=E_a, \quad \quad \tr(F_{a|\vec{r}})=q_{\vec{r}}\ r_a,\\
			&	\quad\sum_a F_{a|\vec{r}}=q_{\vec{r}}\openone, \quad\text{and}\quad F_{a|\vec{r}}\succeq 0.
		\end{aligned}
	\end{equation}
	which is an SDP. Note that the second constraint implies that the distribution over the rank-vectors is given by $q_{\vec{r}}=\frac{1}{d}\sum_a \tr(F_{a|\vec{r}})$. In addition, the SDP can be adapted to compute upper bounds on the projective simulability threshold visibility $v^*(\mathbf{E})$. For this, we need only to introduce the visibility, $v$, as the objective of the SDP, replace $E_a$ with $\Phi_v(E_a)$ in the constraints and run a maximisation of the objective. For qubits and qutrits \eqref{SDPcond} reduces to the one presented in \cite{Oszmaniec2017}. Our implementation is available at \cite{github-code}.
	
	Moreover, sometimes the necessary condition \eqref{SDPcond} is also sufficient for projective simulation. This is the case for any qubit or qutrit POVM. To see this, recall that a projective measurement cannot have more than $d$ non-zero outcomes. Firstly, all  extremal dichotomic measurements are projective \cite{Davies1970}. Secondly, $d$-outcome measurements are a strict superset of the measurements defined via the criterion  \eqref{SDPcond}. This implies that it is exact for qubits. Moreover, it was shown in Ref.~\cite{Oszmaniec2017} that qutrit measurements with unit-trace measurement operators only have projective extremals. Hence, optimality follows also for qutrits.

	Another relevant use of the SDP criterion is to determine witnesses for $\mathbf{E}\notin \mathcal{P}$. These witnesses come in the form of linear inequalities, satisfied by all $\mathbf{E}\in\mathcal{P}$, in the probabilities obtained from probing the measurement with a few selected preparations. These inequalities can in general be extracted from the dual formulation of the SDP; see the Supplemental Material \cite{supplemental_material}. However, we also propose a simple  witness ansatz. Consider a rank-one non-projective measurement $\mathbf{\tilde{E}}$ where $\tilde{E}_a=c_a \ketbra{\psi_a}{\psi_a}$. To show that an arbitrary POVM $\mathbf{E}\notin \mathcal{P}$, we consider a witness inspired by quantum state discrimination,
	\begin{equation}\label{witness}
		W_\mathbf{\tilde{E}}(\mathbf{E})=\frac{1}{d}\sum_{a=1}^n  \bracket{\psi_a}{E_a}{\psi_a},
	\end{equation}
	i.e.~we probe the measurement with the eigenstates of the target POVM with a uniform prior $\frac{1}{d}$. Note that the witness is upper bounded by 1 for any given $\mathbf{E}$,  $\text{max} \frac{1}{d}\sum_{a}\tr\left( \psi_aE_a\right)\leq \frac{1}{d} \sum_{a} \lambda_{\text{max}}(E_a) \leq \frac{1}{d}\sum_{a} \tr\left( E_a\right) =1$.  Moreover, for $\mathbf{E} = \mathbf{\tilde{E}}$ we have  $W_\mathbf{\tilde{E}}(\mathbf{\tilde{E}})= 1$. Thus, we would expect that  $\mathbf{E}$ attains a value close to unit. To show that this is sufficient to elude projective simulation we need to determine the largest value of $	W_\mathbf{\tilde{E}}$ over all projective measurements. This can achieved by SDP, by using \eqref{witness} as the objective over the constraints in Eq~\eqref{SDPcond}. In the Supplemental Material \cite{supplemental_material}, we discuss how to leverage the linearity of $W_{\mathbf{\tilde{E}}}$ to evaluate the SDP more efficiently.

	\textit{Main problem.---} For a given dimension, $d$, it is natural to ask which POVM most strongly defies a projective simulation. This corresponds to finding the specific $\mathbf{E}$ that solves the problem 
	\begin{equation}\label{mainprob}
		v_d\equiv \min_{\mathbf{E}} v^*(\mathbf{E}).
	\end{equation}
	Here, $v_d$ is the visibility associated with the most non-projective measurement in the given dimension. 
	
	A natural candidate for solving this problem are SIC-POVMs. A POVM over $d$-dimensional Hilbert space is SIC if and only if $E_a=\frac{1}{d}\ketbra{\psi_a}{\psi_a}$ where $\left|\braket{\psi_a}{\psi_b}\right|^2=\frac{1}{d+1}$ for $a\neq b$. Such a measurement has $n=d^2$ outcomes, which is the largest possible number for an extremal POVM \cite{DAriano2005}.  All known SIC-POVMs (with one exception for $d=8$) are generated by the Weyl-Heisenberg  group \cite{Fuchs2017}. Define the group's generators  $X=\sum_{j=0}^{d-1} \ketbra{j\oplus 1}{j}$ and $Z=\sum_{j=0}^{d-1}\omega^{j}\ketbra{j}{j}$, where $\omega=e^{2\pi i/d}$. Then, a SIC-POVM is constructed from the orbit
	\begin{equation}\label{sicform}
		\ket{\psi_a}=X^{a_0}Z^{a_1}\ket{\varphi},
	\end{equation}
	where $\ket{\varphi}$ is a suitably chosen vector called a SIC-fiducial and where $a=(a_0,a_1)\in\{0,\ldots,d-1\}^2$.
	
	It was conjectured in Ref.~\cite{Oszmaniec2017} and proven in Ref.~\cite{Hirsch2017} that for $d=2$, the unique SIC-POVM $\mathbf{E}_\text{SIC2}$ is the solution to \eqref{mainprob}. It achieves the visibility $v_2=\sqrt{\frac{2}{3}}\approx 81.6\%$. To achieve this simulation, one requires six rank-one projective measurements corresponding to rank-vectors that are permutations of $(1,1,0,0)$. The corresponding Bloch vectors form a cuboctahedron on the Bloch sphere. It was further conjectured in Ref.~\cite{Oszmaniec2017} that SIC-POVMs also are the solution for $d>2$.

	\textit{Qutrit SIC-POVMs.---}  For qutrits, it is known that all SIC-POVMs admit the form \eqref{sicform} \cite{Szollosi2014, Hughston2016}. Also,  up to a unitary or antiunitary transformation, they can all be obtained from the one-parameter family of SIC-fiducials 
	\begin{equation}\label{eq:sic3fiducial}
		\ket{\varphi^{(\theta)}}=\frac{1}{\sqrt{2}}\left(\ket{1}-e^{i\theta}\ket{2}\right)
	\end{equation}
	for $\theta\in[0,\pi/9]$ \cite{Zhu2010}. Denoting these SIC-POVMs by $\mathbf{E}_\theta$, we use the SDP criterion to compute $v^*(\mathbf{E}_\theta)$. The results are plotted in Fig.~\ref{fig_sic3}. Since the curve is not flat, it means that the SIC-property is insufficient to determine the degree of projective simulability of the POVMs. This stems from the fact that while the SIC-property determines all the magnitudes of the overlaps between the POVM-elements, it does not determine their relative phases. Nevertheless, from the plot we see that every qutrit SIC-POVM has a projective simulability threshold visibility lower than that of the qubit SIC-POVM. In particular, the specific SIC-POVM that attains the smallest visibility ($\theta=0$) is unique and known as the Hesse SIC \cite{Hesse1844}. Its nine (unnormalised) vectors $\{\ket{\psi_a}\}_a$ are the columns of the matrix 
	\begin{equation}\label{Hesse}
		\begin{pmatrix}
			0 & 0 & 0 &-1&-\omega^2&-\omega&1&\omega& \omega^2\\
			1 &\omega & \omega^2&0&0&0&-1&-\omega^2&-\omega\\
			-1& -\omega^2& -\omega&1&\omega&\omega^2&0&0&0\\
		\end{pmatrix}
	\end{equation}
	and it leads to 
	\begin{equation}\label{v3}
		v_3\leq v^*(\mathbf{E}_\text{Hesse})=\frac{1}{6}\left(1+4\cos\left(\frac{\pi}{9}\right)\right)\approx  79.3\%.
	\end{equation}

	\begin{figure}[t!]
		\centering
		\includegraphics[width=0.95\columnwidth]{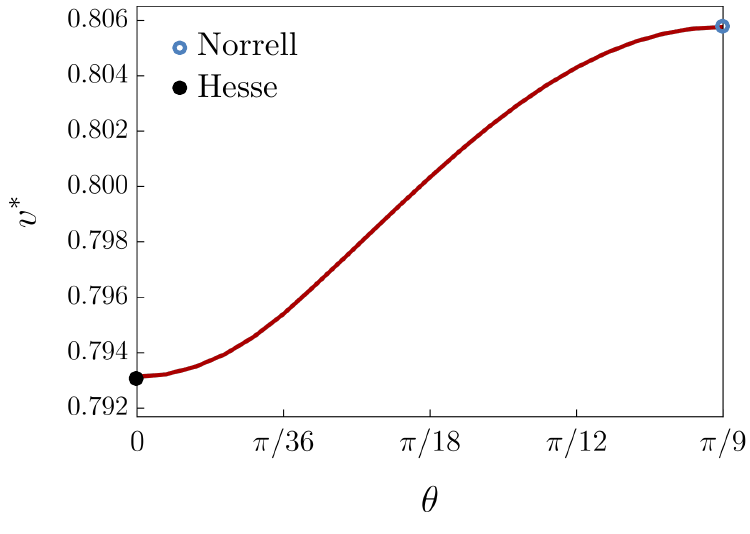}
		\caption{The projective simulability threshold visibility $v^*(\mathbf{E}_\theta)$ defined in Eq.~\eqref{eq:vstar}, for projective simulation of depolarised qutrit SIC-POVMs constructed from fiducial vector $\ket{\varphi^{(\theta)}}$ in Eq.~\eqref{eq:sic3fiducial}.} \label{fig_sic3}
	\end{figure}

	To better understand the projective simulation of the Hesse SIC-POVM, we have found an analytical model that gives the  closed expression above for the decimal solution returned by our SDP criterion. Our simulation strategy requires 72 equiprobable projective measurements. The main idea is to consider the projective Clifford group, which has 216 elements, and compute the orbit of a suitably chosen fiducial vector. The resulting $216$ vectors can be grouped into $72$ sets of $3$ vectors, such that each set forms an orthonormal basis, i.e.~a projective measurement. The model is detailed in the Supplemental Material \cite{supplemental_material}. In contrast, the SIC-POVM with highest visibility threshold ($\theta = \pi/9$) is known in the literature as Norrell SIC and is related to resourceful states in the resource theory of stabilizer computation \cite{Veitch_2014}.

	\textit{Ququart SIC-POVMs.---} In contrast to  three dimensions, there exists up to unitary and antiunitary transformations only a single SIC-POVM in dimension four \cite{Zhu2010_d4}. This POVM, $\mathbf{E}_{\text{SIC4}}$, is obtained from \eqref{sicform} by selecting the appropriate SIC-fiducial vector\footnote{The fiducial is $\ket{\varphi_4} \propto \sqrt{2}\ket{0} + [z(1-z)(\phi^{3/2}+ \bar{z})]\ket{1} +(2-\sqrt{2})i\ket{2} + [z(1-z)(\phi^{3/2}- \bar{z})]\ket{3}$ with $z = e^{i\pi/4}$ and the golden ratio $\phi = (1+\sqrt{5})/2$.}.
	In contrast to the lower-dimensional cases, the value obtained by the SDP yields only an upper bound on the projective simulability threshold visibility. Remarkably, however, the bound is exact for the SIC-POVM and we obtain $v^*(\mathbf{E}_{\text{SIC4}}) \approx 82.6\%$. It follows from the fact that the solution returned by our SDP has a particular structure: all the non-zero operators take the form $F_{a|\vec{r}}=q_{\vec{r}}\ O_{a|\vec{r}}$ for some projective operator $O_{a|\vec{r}}$. From \eqref{SDPcond} we have that $\sum_a F_{a|\vec{r}}=q_{\vec{r}}\openone$ which then implies that $\sum_a O_{a|\vec{r}}=\openone$. Thus, $\{O_{a|\vec{r}}\}_a$ is a projective measurement for every relevant $\vec{r}$. Hence, we have an explicit projective strategy that saturates our upper bound, thus implying that it is optimal. Our simulation relies on $132$  projective measurements.

	\textit{The most non-projective POVM.---} We now set out to investigate which are the most non-projective measurements beyond qubit systems. To study this problem, we have used a numerical search algorithm. This algorithm uses the SDP criterion for projective simulability as an oracle to iteratively explore the space of non-projective measurements. Specifically, its steps are the following. (i) Select a non-projective measurement from which to begin the search and use the SDP criterion to compute a bound on its visibility. (ii) From the dual SDP (see the Supplemental Material \cite{supplemental_material}), extract a witness for genuine non-projectivity. (iii) Compute the measurement that gives the largest violation of this witness. This can also be done by an SDP. (iv) Return to step (i) and use the measurement from (iii) as the input. This process is repeated until desired convergence is found. Our implementation is available at \cite{github-code}. As a basic check, we have run this process first for qubits and find that it essentially always converges to the SIC-POVM, which we know is the correct answer. 
	
	For qutrits, we have evaluated this search with many different initial POVMs and found that it has, essentially always, two possible fates: either it gets stuck in a specific local extremum or it bypasses it to converge to the Hesse SIC-POVM \eqref{Hesse}. This provides significant evidence that Eq~\eqref{v3} actually is an equality and that the Hesse SIC-POVM is the most non-projective qutrit measurement. Furthermore, also the POVM found in the local extremum is interesting, because it has only five (out of nine possible) outcomes. In fact, it turns out to closely resemble the qubit SIC-POVM, $\mathbf{E}_\text{SIC2}$. Its POVM elements are
	\begin{align}\label{fsic2}
		& \{E_1,E_2,E_3,E_4\}=\mathbf{E}_\text{SIC2}\oplus 0, && E_5 =\ketbra{2}{2}.
	\end{align} 
	Thus, the first four outcomes are just an embedding of the qubit SIC-POVM and the fifth outcome is a vector orthogonal to the qubit sub-space. Therefore, we call this a flagged SIC-POVM, $\mathbf{E}_{\text{fSIC2}}$. Its visibility is $v^*(\mathbf{E}_\text{fSIC2})\approx 79.6\%$, which is higher than for the Hesse SIC-POVM but lower than for the qubit SIC-POVM. Hence, simply by appending an orthogonal projection  (``flag''), in contrast to intuition, the non-projective features of the qubit SIC-POVM are amplified. In the Supplemental Material \cite{supplemental_material} we construct an analytical projective simulation of the noisy $\mathbf{E}_\text{fSIC2}$, as well as its generalisations to higher dimension, which shows that it exhibits a monotonically decreasing simulability threshold with increasing dimension. For example, the qubit SIC flagged in $d=4$ yields a threshold visibility of $78.6\%$, which is lower than $v^{*}(\mathbf{E}_{\text{SIC4}})$. This shows that $\mathbf{E}_{\text{SIC4}}$ cannot be the most non-projective measurement in $d=4$.

	\begin{figure}[t!]
		\centering
		\includegraphics[width=0.95\columnwidth]{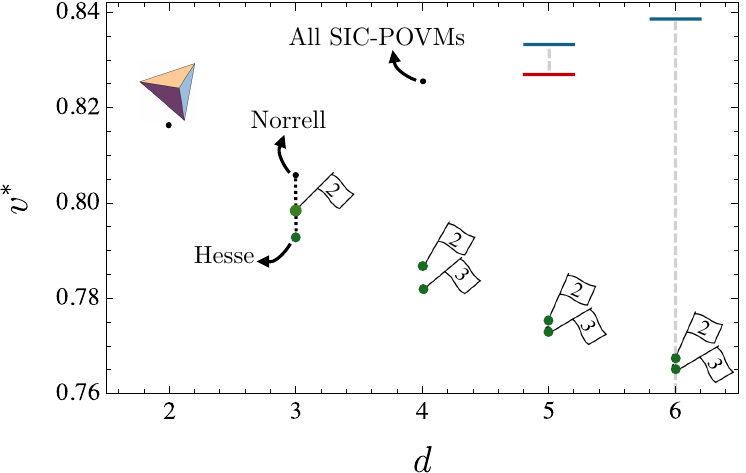}
		\caption{The threshold visibility $v^*$ from Eq.~\eqref{eq:vstar} obtained for various POVMS in $2\leq d\leq 6$. Black dots represent SIC-POVMs, while green ones fSIC-POVMs (the number in the flag corresponds to the dimension of the embedded SIC-POVM), for which the visibility thresholds are calculated in \cite{supplemental_material}. The blue lines represent upper bounds on the visibility obtained for SIC-POVMs in $d=5,6$ using the witness method, while the red line denotes a lower bound in $d=5$ found by explicit construction.}\label{fig_sicd}
	\end{figure}
    
	Therefore, we have repeated the numerical search also for ququarts. The best search result returns that $v_4\lesssim 78.2\%$. Interestingly, the corresponding POVM does not use all the possible 16 outcomes, but only ten of them. Indeed, this one too has a flagged structure: it is an embedding of the Hesse SIC-POVM appended with an orthogonal projection,
	\begin{align}\label{fsic3}
		\{E_1\ldots,E_9\}=\mathbf{E}_\text{Hesse}\oplus 0, && E_{10} =\ketbra{3}{3}.
	\end{align}
	We call this POVM $\mathbf{E}_\text{fSIC3}$. In analogy with the qutrit case, we construct a projective simulation of the noisy $\mathbf{E}_\text{fSIC3}$ (and its higher-dimensional generalizations) in the Supplemental Material \cite{supplemental_material} which confirms $v^*(\mathbf{E}_\text{fSIC3}) \approx 78.2\%$. The visibility of this POVM is lower than that of the SIC-POVM embedded into it and we conjecture that this is the most non-projective ququart measurement.

	\textit{Discussion.---} We have shown that SIC-POVMs in general are not the POVMs with the largest depolarisation noise-robustness with respect to projective simulations. For qutrits, we have shown that not all SIC-POVMs have the same robustness and for ququarts we have shown that SIC-POVMs are far from the most non-projective measurements. Nevertheless, our conjectures on the qutrit and ququart measurements that are the most non-projective clearly relate to SIC-inspired features.
	
	In particular, contrary to what we initially expected, we have found that the robustness of SIC-POVMs does not always decrease as the Hilbert space dimension gets larger. In Fig.~\ref{fig_sicd} we showcase this feature, while the quantitative results are listed in the Supplemental Material \cite{supplemental_material}. In this figure, we have also used the witness ansatz \eqref{witness} to compute upper bounds on the projective simulability threshold of SIC-POVMs in dimensions five and six. Notably, this method returns optimal results for $d=2$ and the Hesse SIC but somewhat suboptimal results for $d=4$. In $d=5$, by explicit construction we have shown that a projective simulation is possible up to $v\approx 82.7\%$ (see \cite{github-code}). Thus, the fact that SIC-POVMs are not the most non-projective measurements does not seem to be related to whether the dimension is prime or not. Moreover, we observe increasing simulability thresholds from dimensions $d=3$ to $d=5$. Whether this truly reflects a decaying non-projective property of SIC-POVMs with increasing dimension is an interesting problem beyond the scope of this work.

	So far, we limited our attention to the case of depolarising noise. A complementary model that is often used in more mathematically motivated literature on resource theories of measurements is derived from considering maximally detrimental measurements  \cite{Chitambar_2019}. This worst-case noise has well-behaved mathematical properties \cite{Schluck_2023}. It corresponds to the situation in which the target POVM, $\mathbf{E}$, is mixed with another arbitrary POVM, $\mathbf{E}_\text{noise}$. By this measure, the projective simulation threshold visibility of $\mathbf{E}$ is the largest $v$ such that one can find  some $\mathbf{E}_\text{noise}$ that satisfies $v\mathbf{E}+(1-v)\mathbf{E}_\text{noise}\in\mathcal{P}$ (see the Supplemental Material \cite{supplemental_material}); we refer to this as the ``worst-case'' noise. For example, if $\mathbf{E}=\mathbf{E}_\text{SIC2}$ then this ``worst-case'' noise will be another SIC-POVM to which $\mathbf{E}_\text{SIC2}$ is elementwise orthogonal. As shown in the Supplemental Material \cite{supplemental_material}, the corresponding  threshold visibility for any fixed POVM cannot decrease when adding flag dimensions to it, thus, it does not suffer from the unintuitive behaviour that we observe for $\mathbf{E}_\text{fSIC2}$ and $\mathbf{E}_\text{fSIC3}$, namely that they get more robust to noise if additional trivial flags are added to it.
	
	For the sake of completeness, we have used our SDP criterion to analyse the problem of finding the most non-projective measurement also under this global measure. The results are qualitatively analogous to those for depolarisation noise. For qubits and qutrits we find that the same SIC-POVMs are optimal. What is more, for ququarts, we again find that the SIC-POVM is not the most non-projective measurement and that the flagged Hesse SIC-POVM \eqref{fsic3} performs better. See the Supplemental Material \cite{supplemental_material} for quantitative results. This indicates that the main results reported here are not artefacts unique to the choice of depolarisation as the quantifier of non-projective properties.
	
	Our findings open several paths for future investigations. First, the question arises whether the results can be extended to the case of quantum instruments, for which a similar concept of simulability was introduced recently \cite{Khandelwal_2025}. Second, although our work provides crucial numerical evidence that the flagged Hesse SIC is the maximally robust POVM in four dimensions, it is still left to prove this conjecture. Moreover, while the maximally non-projective POVM for qubits (i.e., the SIC-POVM) provides advantages over other POVMs in several tasks, it would be interesting to show the same for maximally non-projective POVMs in other dimensions.

	\begin{acknowledgments}
		We thank Ingemar Bengtsson, Markus Grassl and Joseph Renes for sharing their knowledge about SIC-POVMs and  David Gross, Tulja Varun Kondra, Lucas Tendick, Roope Uola, Malte Zehnpfennig, and Amro Abou-Hachem for fruitful discussions. This work is supported by the Swedish Research Council under Contract No. 2023-03498 and the Knut and Alice Wallenberg Foundation through the Wallenberg Center for Quantum Technology (WACQT). S.K.~acknowledges support from the Swiss
		National Science Foundation Grant No.~P500PT\_222265. R.B., H.K.~and D.B.~acknowledge support by Deutsche Forschungsgemeinschaft (DFG, German Research Foundation)
		under Germany’s Excellence Strategy -- Cluster of Excellence Matter and Light for Quantum Computing (ML4Q) EXC 2004/1 -- 390534769. N.W.~acknowledges support by EIN Quantum NRW. 
	\end{acknowledgments}
	
	\bibliography{references_manuscript}

@article{DAriano2005,
doi = {10.1088/0305-4470/38/26/010},
url = {https://dx.doi.org/10.1088/0305-4470/38/26/010},
year = {2005},
month = {jun},
publisher = {},
volume = {38},
number = {26},
pages = {5979},
author = {D'Ariano, Giacomo Mauro and Lo Presti, Paoloplacido and Perinotti, Paolo},
title = {Classical randomness in quantum measurements},
journal = {J. Phys. A: Math. Gen.}
}

@article{Renes2004,
	author = {Renes, Joseph M. and Blume-Kohout, Robin and Scott, A. J. and Caves, Carlton M.},
	title = {Symmetric informationally complete quantum measurements},
	journal = {J. Math. Phys.},
	volume = {45},
	number = {6},
	pages = {2171-2180},
	year = {2004},
	month = {06},
	issn = {0022-2488},
	doi = {10.1063/1.1737053},
	url = {https://doi.org/10.1063/1.1737053},
}

@article{Shang2018,
	title = {Enhanced entanglement criterion via symmetric informationally complete measurements},
	author = {Shang, Jiangwei and Asadian, Ali and Zhu, Huangjun and G\"uhne, Otfried},
	journal = {Phys. Rev. A},
	volume = {98},
	issue = {2},
	pages = {022309},
	numpages = {7},
	year = {2018},
	month = {Aug},
	publisher = {American Physical Society},
	doi = {10.1103/PhysRevA.98.022309},
	url = {https://link.aps.org/doi/10.1103/PhysRevA.98.022309}
}

@article{Wang2023,
	title = {Generalized Quantum Measurements on a Higher-Dimensional System via Quantum Walks},
	author = {Wang, Xiaowei and Zhan, Xiang and Li, Yulin and Xiao, Lei and Zhu, Gaoyan and Qu, Dengke and Lin, Quan and Yu, Yue and Xue, Peng},
	journal = {Phys. Rev. Lett.},
	volume = {131},
	issue = {15},
	pages = {150803},
	numpages = {7},
	year = {2023},
	month = {Oct},
	publisher = {American Physical Society},
	doi = {10.1103/PhysRevLett.131.150803},
	url = {https://link.aps.org/doi/10.1103/PhysRevLett.131.150803}
}

@Article{Hou2018,
	author={Hou, Zhibo
	and Tang, Jun-Feng
	and Shang, Jiangwei
	and Zhu, Huangjun
	and Li, Jian
	and Yuan, Yuan
	and Wu, Kang-Da
	and Xiang, Guo-Yong
	and Li, Chuan-Feng
	and Guo, Guang-Can},
	title={Deterministic realization of collective measurements via photonic quantum walks},
	journal={Nat. Commun.},
	year={2018},
	month={Apr},
	day={12},
	volume={9},
	number={1},
	pages={1414},
	issn={2041-1723},
	doi={10.1038/s41467-018-03849-x},
	url={https://doi.org/10.1038/s41467-018-03849-x}
}

@article{Zhao2015,
	title = {Experimental realization of generalized qubit measurements based on quantum walks},
	author = {Zhao, Yuan-yuan and Yu, Neng-kun and Kurzy\ifmmode \acute{n}\else \'{n}\fi{}ski, Pawe\l{} and Xiang, Guo-yong and Li, Chuan-Feng and Guo, Guang-Can},
	journal = {Phys. Rev. A},
	volume = {91},
	issue = {4},
	pages = {042101},
	numpages = {7},
	year = {2015},
	month = {Apr},
	publisher = {American Physical Society},
	doi = {10.1103/PhysRevA.91.042101},
	url = {https://link.aps.org/doi/10.1103/PhysRevA.91.042101}
}

@article{Bian2015,
	title = {Realization of Single-Qubit Positive-Operator-Valued Measurement via a One-Dimensional Photonic Quantum Walk},
	author = {Bian, Zhihao and Li, Jian and Qin, Hao and Zhan, Xiang and Zhang, Rong and Sanders, Barry C. and Xue, Peng},
	journal = {Phys. Rev. Lett.},
	volume = {114},
	issue = {20},
	pages = {203602},
	numpages = {5},
	year = {2015},
	month = {May},
	publisher = {American Physical Society},
	doi = {10.1103/PhysRevLett.114.203602},
	url = {https://link.aps.org/doi/10.1103/PhysRevLett.114.203602}
}

@article{Mironowicz2019,
	title = {Experimentally feasible semi-device-independent certification of four-outcome positive-operator-valued measurements},
	author = {Mironowicz, Piotr and Paw\l{}owski, Marcin},
	journal = {Phys. Rev. A},
	volume = {100},
	issue = {3},
	pages = {030301},
	numpages = {6},
	year = {2019},
	month = {Sep},
	publisher = {American Physical Society},
	doi = {10.1103/PhysRevA.100.030301},
	url = {https://link.aps.org/doi/10.1103/PhysRevA.100.030301}
}

@article{Hughston2016,
	title = {Surveying points in the complex projective plane},
	journal = {Adv. Math.},
	volume = {286},
	pages = {1017-1052},
	year = {2016},
	issn = {0001-8708},
	doi = {https://doi.org/10.1016/j.aim.2015.09.022},
	url = {https://www.sciencedirect.com/science/article/pii/S0001870815003631},
	author = {Lane P. Hughston and Simon M. Salamon}
}

@article{Hesse1844,
	url = {https://doi.org/10.1515/crll.1844.28.68},
	title = {{Ü}ber die {E}limination der {V}ariabeln aus drei algebraischen {G}leichungen vom zweiten {G}rade mit zwei {V}ariabeln},
	author = {Otto Hesse},
	pages = {68--96},
	volume = {1844},
	number = {28},
	journal = {Journal für die reine und angewandte Mathematik},
	doi = {doi:10.1515/crll.1844.28.68},
	year = {1844},
	lastchecked = {2025-07-16}
}

@article{Zhu2010,
	doi = {10.1088/1751-8113/43/30/305305},
	url = {https://dx.doi.org/10.1088/1751-8113/43/30/305305},
	year = {2010},
	month = {jun},
	publisher = {},
	volume = {43},
	number = {30},
	pages = {305305},
	author = {Zhu, Huangjun},
	title = {{SIC} {POVMs} and Clifford groups in prime dimensions},
	journal = {J. Phys. A: Math. Theor.}
}

@article{Zauner1999,
	title={Grundz{\"u}ge einer nichtkommutativen {D}esigntheorie},
	author={Zauner, Gerhard},
	journal={Ph. D. dissertation, PhD thesis},
	year={1999},
	publisher={University of Vienna}
}

@article{Szollosi2014,
	title={All complex equiangular tight frames in dimension 3}, 
	author={Ferenc Sz\"{o}ll\"{o}si},
	year={2014},
journal={arXiv:1402.6429},
	url={https://arxiv.org/abs/1402.6429}, 
}

@article{Tavakoli2024b,
	title = {Semidefinite programming relaxations for quantum correlations},
	author = {Tavakoli, Armin and Pozas-Kerstjens, Alejandro and Brown, Peter and Ara\'ujo, Mateus},
	journal = {Rev. Mod. Phys.},
	volume = {96},
	issue = {4},
	pages = {045006},
	numpages = {68},
	year = {2024},
	month = {Dec},
	publisher = {American Physical Society},
	doi = {10.1103/RevModPhys.96.045006},
	url = {https://link.aps.org/doi/10.1103/RevModPhys.96.045006}
}

@article{Feng2025,
	author = {Lan-Tian Feng and Xiao-Min Hu and Ming Zhang and Yu-Jie Cheng and Chao Zhang and Yu Guo and Yu-Yang Ding and Zhibo Hou and Fang-Wen Sun and Guang-Can Guo and Dao-Xin Dai and Armin Tavakoli and Xi-Feng Ren and Bi-Heng Liu},
	journal = {Optica},
	number = {7},
	pages = {1014--1019},
	publisher = {Optica Publishing Group},
	title = {Higher-dimensional symmetric informationally complete measurement via programmable photonic integrated optics},
	volume = {12},
	month = {Jul},
	year = {2025},
	url = {https://opg.optica.org/optica/abstract.cfm?URI=optica-12-7-1014},
	doi = {10.1364/OPTICA.551264},
}

@article{Fuchs2017,
	AUTHOR = {Fuchs, Christopher A. and Hoang, Michael C. and Stacey, Blake C.},
	TITLE = {The {SIC} Question: History and State of Play},
	JOURNAL = {Axioms},
	VOLUME = {6},
	YEAR = {2017},
	NUMBER = {3},
	ARTICLE-NUMBER = {21},
	URL = {https://www.mdpi.com/2075-1680/6/3/21},
}

@article{Smania2020,
	author = {Armin Tavakoli  and Massimiliano Smania  and Tamás Vértesi  and Nicolas Brunner  and Mohamed Bourennane },
	title = {Self-testing nonprojective quantum measurements in prepare-and-measure experiments},
	journal = {Sci. Adv.},
	volume = {6},
	number = {16},
	pages = {eaaw6664},
	year = {2020},
	doi = {10.1126/sciadv.aaw6664},
	URL = {https://www.science.org/doi/abs/10.1126/sciadv.aaw6664}}

@article{Horodecki2022,
	title = {Five Open Problems in Quantum Information Theory},
	author = {Horodecki, Pawe\l{} and Rudnicki, \L{}ukasz and \ifmmode \dot{Z}\else \.{Z}\fi{}yczkowski, Karol},
	journal = {PRX Quantum},
	volume = {3},
	issue = {1},
	pages = {010101},
	numpages = {17},
	year = {2022},
	month = {Mar},
	publisher = {American Physical Society},
	doi = {10.1103/PRXQuantum.3.010101},
	url = {https://link.aps.org/doi/10.1103/PRXQuantum.3.010101}
}

@article{Tavakoli2024a,
	title = {Enhanced Schmidt-number criteria based on correlation trace norms},
	author = {Tavakoli, Armin and Morelli, Simon},
	journal = {Phys. Rev. A},
	volume = {110},
	issue = {6},
	pages = {062417},
	numpages = {9},
	year = {2024},
	month = {Dec},
	publisher = {American Physical Society},
	doi = {10.1103/PhysRevA.110.062417},
	url = {https://link.aps.org/doi/10.1103/PhysRevA.110.062417}
}

@Article{Martinez2023,
	author={Mart{\'i}nez, Daniel
	and G{\'o}mez, Esteban S.
	and Cari{\~{n}}e, Jaime
	and Pereira, Luciano
	and Delgado, Aldo
	and Walborn, Stephen P.
	and Tavakoli, Armin
	and Lima, Gustavo},
	title={Certification of a non-projective qudit measurement using multiport beamsplitters},
	journal={Nat. Phys.},
	year={2023},
	month={Feb},
	day={01},
	volume={19},
	number={2},
	pages={190-195},
	issn={1745-2481},
	doi={10.1038/s41567-022-01845-z},
	url={https://doi.org/10.1038/s41567-022-01845-z}
}

@Article{Singal2022,
	author={Singal, Tanmay
	and Maciejewski, Filip B.
	and Oszmaniec, Micha{\l}},
	title={Implementation of quantum measurements using classical resources and only a single ancillary qubit},
	journal={npj Quant. Inf.},
	year={2022},
	month={Jul},
	day={13},
	volume={8},
	number={1},
	pages={82},
	issn={2056-6387},
	doi={10.1038/s41534-022-00589-1},
	url={https://doi.org/10.1038/s41534-022-00589-1}
}

@article{Guerini2017,
	author = {Guerini, Leonardo and Bavaresco, Jessica and Terra Cunha, Marcelo and Acín, Antonio},
	title = {Operational framework for quantum measurement simulability},
	journal = {J. Math. Phys.},
	volume = {58},
	number = {9},
	pages = {092102},
	year = {2017},
	month = {09},
	issn = {0022-2488},
	doi = {10.1063/1.4994303},
	url = {https://doi.org/10.1063/1.4994303},
}

@Article{Davies1970,
	author={Davies, E. B.
	and Lewis, J. T.},
	title={An operational approach to quantum probability},
	journal={Commun. Math. Phys.},
	year={1970},
	month={Sep},
	day={01},
	volume={17},
	number={3},
	pages={239-260},
	issn={1432-0916},
	doi={10.1007/BF01647093},
	url={https://doi.org/10.1007/BF01647093}
}

@article{Hirsch2017,
	doi = {10.22331/q-2017-04-25-3},
	url = {https://doi.org/10.22331/q-2017-04-25-3},
	title = {Better local hidden variable models for two-qubit {W}erner states and an upper bound on the {G}rothendieck constant {$K_G(3)$}},
	author = {Hirsch, Flavien and Quintino, Marco T{\'{u}}lio and V{\'{e}}rtesi, Tam{\'{a}}s and Navascu{\'{e}}s, Miguel and Brunner, Nicolas},
	journal = {{Quantum}},
	issn = {2521-327X},
	publisher = {{Verein zur F{\"{o}}rderung des Open Access Publizierens in den Quantenwissenschaften}},
	volume = {1},
	pages = {3},
	month = apr,
	year = {2017}
}

@article{Tavakoli2019,
	title = {Enabling Computation of Correlation Bounds for Finite-Dimensional Quantum Systems via Symmetrization},
	author = {Tavakoli, Armin and Rosset, Denis and Renou, Marc-Olivier},
	journal = {Phys. Rev. Lett.},
	volume = {122},
	issue = {7},
	pages = {070501},
	numpages = {7},
	year = {2019},
	month = {Feb},
	publisher = {American Physical Society},
	doi = {10.1103/PhysRevLett.122.070501},
	url = {https://link.aps.org/doi/10.1103/PhysRevLett.122.070501}
}

@article{Stricker2022,
	title = {Experimental Single-Setting Quantum State Tomography},
	author = {Stricker, Roman and Meth, Michael and Postler, Lukas and Edmunds, Claire and Ferrie, Chris and Blatt, Rainer and Schindler, Philipp and Monz, Thomas and Kueng, Richard and Ringbauer, Martin},
	journal = {PRX Quantum},
	volume = {3},
	issue = {4},
	pages = {040310},
	numpages = {34},
	year = {2022},
	month = {Oct},
	publisher = {American Physical Society},
	doi = {10.1103/PRXQuantum.3.040310},
	url = {https://link.aps.org/doi/10.1103/PRXQuantum.3.040310}
}

@article{Gisin2019,
	AUTHOR = {Gisin, Nicolas},
	TITLE = {Entanglement 25 Years after Quantum Teleportation: Testing Joint Measurements in Quantum Networks},
	JOURNAL = {Entropy},
	VOLUME = {21},
	YEAR = {2019},
	NUMBER = {3},
	ARTICLE-NUMBER = {325},
	URL = {https://www.mdpi.com/1099-4300/21/3/325},
}

@article{Fuchs2013,
	title = {Quantum-Bayesian coherence},
	author = {Fuchs, Christopher A. and Schack, R\"udiger},
	journal = {Rev. Mod. Phys.},
	volume = {85},
	issue = {4},
	pages = {1693--1715},
	numpages = {0},
	year = {2013},
	month = {Dec},
	publisher = {American Physical Society},
	doi = {10.1103/RevModPhys.85.1693},
	url = {https://link.aps.org/doi/10.1103/RevModPhys.85.1693}
}

@article{Acin2016,
	title = {Optimal randomness certification from one entangled bit},
	author = {Ac\'{\i}n, Antonio and Pironio, Stefano and V\'ertesi, Tam\'as and Wittek, Peter},
	journal = {Phys. Rev. A},
	volume = {93},
	issue = {4},
	pages = {040102},
	numpages = {5},
	year = {2016},
	month = {Apr},
	publisher = {American Physical Society},
	doi = {10.1103/PhysRevA.93.040102},
	url = {https://link.aps.org/doi/10.1103/PhysRevA.93.040102}
}

@article{Huang2021,
	title = {Nonlocality, Steering, and Quantum State Tomography in a Single Experiment},
	author = {Huang, Chang-Jiang and Xiang, Guo-Yong and Guo, Yu and Wu, Kang-Da and Liu, Bi-Heng and Li, Chuan-Feng and Guo, Guang-Can and Tavakoli, Armin},
	journal = {Phys. Rev. Lett.},
	volume = {127},
	issue = {2},
	pages = {020401},
	numpages = {6},
	year = {2021},
	month = {Jul},
	publisher = {American Physical Society},
	doi = {10.1103/PhysRevLett.127.020401},
	url = {https://link.aps.org/doi/10.1103/PhysRevLett.127.020401}
}

@article{Bengtsson2012,
	title = {A {K}ochen–{S}pecker inequality from a SIC},
	journal = {Phys. Lett. A},
	volume = {376},
	number = {4},
	pages = {374-376},
	year = {2012},
	issn = {0375-9601},
	doi = {https://doi.org/10.1016/j.physleta.2011.12.011},
	url = {https://www.sciencedirect.com/science/article/pii/S037596011101454X},
	author = {Ingemar Bengtsson and Kate Blanchfield and Adán Cabello}
}

@article{Scott2006,
	doi = {10.1088/0305-4470/39/43/009},
	url = {https://dx.doi.org/10.1088/0305-4470/39/43/009},
	year = {2006},
	month = {oct},
	publisher = {},
	volume = {39},
	number = {43},
	pages = {13507},
	author = {Scott, A J},
	title = {Tight informationally complete quantum measurements},
	journal = {J. Phys. A: Math. Gen.}
}

@article{Oszmaniec2017,
	title = {Simulating Positive-Operator-Valued Measures with Projective Measurements},
	author = {Oszmaniec, Micha\l{} and Guerini, Leonardo and Wittek, Peter and Ac\'{\i}n, Antonio},
	journal = {Phys. Rev. Lett.},
	volume = {119},
	issue = {19},
	pages = {190501},
	numpages = {6},
	year = {2017},
	month = {Nov},
	publisher = {American Physical Society},
	doi = {10.1103/PhysRevLett.119.190501},
	url = {https://link.aps.org/doi/10.1103/PhysRevLett.119.190501}
}

@book{Skrzypczyk_2023,
   title={Semidefinite Programming in Quantum Information Science},
   ISBN={9780750333412},
   url={http://dx.doi.org/10.1088/978-0-7503-3343-6},
   DOI={10.1088/978-0-7503-3343-6},
   publisher={IOP Publishing},
   author={Skrzypczyk, Paul and Cavalcanti, Daniel},
   year={2023},
   month=mar }

@article{Appleby_2005,
  title={Symmetric informationally complete--positive operator valued measures and the extended Clifford group},
  author={Appleby, D Marcus},
  journal={J. Math. Phys.},
  volume={46},
  number={5},
  year={2005},
  publisher={AIP Publishing},
  url={https://doi.org/10.1063/1.1896384}
}

@article{Scott_2010,
  title={Symmetric informationally complete positive-operator-valued measures: A new computer study},
  author={Scott, Andrew James and Grassl, Markus},
  journal={J. Math. Phys.},
  volume={51},
  number={4},
  year={2010},
  publisher={AIP Publishing},
  url={https://doi.org/10.1063/1.3374022}
}

@article{Chitambar_2019,
  title={Quantum resource theories},
  author={Chitambar, Eric and Gour, Gilad},
  journal={Rev. Mod. Phys.},
  volume={91},
  number={2},
  pages={025001},
  year={2019},
  publisher={APS},
  url={https://doi.org/10.1103/RevModPhys.91.025001}
}

@article{Veitch_2014,
  title={The resource theory of stabilizer quantum computation},
  author={Veitch, Victor and Mousavian, SA Hamed and Gottesman, Daniel and Emerson, Joseph},
  journal={New J. Phys.},
  volume={16},
  number={1},
  pages={013009},
  year={2014},
  publisher={IOP Publishing},
  url={https://doi.org/10.1088/1367-2630/16/1/013009}
}

@article{Schluck_2023,
  title={Continuity of robustness measures in quantum resource theories},
  author={Schluck, Jonathan and Murta, Gl{\'a}ucia and Kampermann, Hermann and Bru{\ss}, Dagmar and Wyderka, Nikolai},
  journal={J. Phys. A: Math. Theor.},
  volume={56},
  number={25},
  pages={255303},
  year={2023},
  publisher={IOP Publishing},
  url={https://doi.org/10.1088/1751-8121/acd500}
}

@article{Khandelwal_2025,
  title={Simulating quantum instruments with projective measurements and quantum postprocessing},
  author={Khandelwal, Shishir and Tavakoli, Armin},
  journal={Phys. Rev. Lett.},
  volume={135},
  number={4},
  pages={040202},
  year={2025},
  publisher={APS},
  url={https://doi.org/10.1103/bhr5-g71p}
}

@article{Webb_2016,
  title={The Clifford group forms a unitary 3-design},
  author={Webb, Zak},
  journal={Quantum Inf. Comput.},
  volume={16},
  pages={1379},
  year={2016},
  publisher={Rinton Press},
  url={https://doi.org/10.26421/QIC16.15-16-8}
}

@misc{github-code,
 	author = {},
 	title = {Code for PVM-simulability and numerical search},
	note = {\url{https://github.com/GabrieleCobucci/POVM_simulation}},
	year = 2025,
}

@misc{kotowski2025prettygoodsimulationquantummeasurements,
      title={Pretty-good simulation of all quantum measurements by projective measurements}, 
      author={Michał Kotowski and Michał Oszmaniec},
      year={2025},
      eprint={2501.09339},
      archivePrefix={arXiv},
      primaryClass={quant-ph},
      url={https://arxiv.org/abs/2501.09339}, 
}

@article{Zhu2010_d4,
   title={Two-qubit symmetric informationally complete positive-operator-valued measures},
   volume={82},
   ISSN={1094-1622},
   url={http://dx.doi.org/10.1103/PhysRevA.82.042308},
   DOI={10.1103/physreva.82.042308},
   number={4},
   journal={Physical Review A},
   publisher={American Physical Society (APS)},
   author={Zhu, Huangjun and Teo, Yong Siah and Englert, Berthold-Georg},
   year={2010},
   month=oct }

@article{liu2021quantifying,
  title = {Quantifying unsharpness of measurements via uncertainty},
  author = {Liu, Yizhou and Luo, Shunlong},
  journal = {Phys. Rev. A},
  volume = {104},
  issue = {5},
  pages = {052227},
  numpages = {10},
  year = {2021},
  month = {Nov},
  publisher = {American Physical Society},
  doi = {10.1103/PhysRevA.104.052227},
  url = {https://link.aps.org/doi/10.1103/PhysRevA.104.052227}
}

@misc{supplemental_material,
      title={See Supplemental Material}, 
}

	\appendix
	
	\section{Redundancy of post-processing in projective simulation}\label{AppProofPostprocessing}
	
	\noindent The projective simulation of a POVM $\mathbf E=\{E_a\}_a$ entails the following decomposition
	\begin{align}\label{eq:simeq}
		E_a = \sum_{\lambda} q_\lambda\sum_k p(a\lvert k,\lambda)P_{k\lvert \lambda} ,
	\end{align}
	where $\lambda$ is a random variable with probability distribution $q_\lambda$, $P_{k\lvert \lambda}$ is a projective measurement for each given $\lambda$  and $p(a\lvert k,\lambda)$ is a probabilistic post-processing rule. In this section, we show that the classical post-processing step in the above simulation is redundant.
	
	Let us first consider the case of deterministic post-processing, i.e., $p(a|k,\lambda) =\delta_{a,f(k,\lambda)} $, where $f$ is some function. Our simulation equation is then
	\begin{align}
		E_a = \sum_{\lambda} q_{\lambda}\sum_k P_{k\lvert \lambda} \delta_{a,f(k,\lambda)}.
	\end{align}We define
	\begin{align}
		F_{a\lvert \lambda} \coloneqq \sum_{k} P_{k\lvert\lambda} \delta_{a,f(k,\lambda)}.
	\end{align}Note that $F_{a\lvert \lambda} $ is implicitly dependent on the function $f$. It can be easily verified that $\{F_{a\lvert \lambda}\}_a$ is a projective measurement for any choice of $(\lambda,f)$. With this definition, we have 
	\begin{equation}
		\begin{aligned}
			E_a &= \sum_{\lambda}q_{\lambda} F_{a\lvert \lambda}
		\end{aligned}
	\end{equation}We have therefore eliminated the post-processing completely and have written the POVM as a stochastic implementation of the projectors $F_{a\lvert \lambda}$.

	Next we consider arbitrary post-processing. Note that such a case can be realised by the stochastic application of deterministic strategies, i.e., $p(a\lvert k,\lambda) = \sum_g  q_g \delta_{a,g(k,\lambda)}$, for some functions $g$. We can therefore write the simulation equation \eqref{eq:simeq} as
	\begin{equation}
		\begin{aligned}
			E_a &= \sum_{\lambda}q_{\lambda}\sum_{k,g} P_{k\lvert\lambda}\, q_g \delta_{a,g(k,\lambda)} \\&= \sum_{\lambda,g} q_{\lambda}q_g\sum_{k} P_{k\lvert\lambda}\delta_{a,g(k,\lambda)} \\ 
			& = \sum_{\lambda,g} q_{\lambda,g} G_{a\lvert \lambda,g},
		\end{aligned}
	\end{equation}where we have combined the two probability distributions and have defined $G_{a\lvert \lambda,g}\coloneqq\sum_{k} P_{k\lvert\lambda}\delta_{a,g(k,\lambda)}$. As before, $\{G_{a\lvert \lambda,g}\}_a$ is a projective measurement for any choice of $(\lambda,g)$. Defining a new random variable $\lambda'\coloneqq (\lambda,g)$, we have
	\begin{align}
		E_a = \sum_{\lambda'}q_{\lambda'}G_{a\lvert\lambda'}.
	\end{align}We have therefore again eliminated the post-processing completely, as above for the deterministic post-processing case.
	
	\qed

	\section{Dual SDP}\label{AppDual}
	In this appendix, we derive the dual formulation of the SDP introduced in the main text. In particular, we consider the case in which the SDP computes upper bounds on the projective simulability threshold visibility $v^{*}(\mathbf{E})$, i.e.
	\begin{equation}\label{SDP_primal_visibility}
		\begin{aligned}
			\max_{v,\{F_{a|\vec{r}}\}_{a,\vec{r}}} & \quad v\\
			\text{s.t.}&\quad v E_a+\frac{1-v}{d}\tr(E_a)\openone =  \sum_{\vec{r}} F_{a|\vec{r}},  \quad \forall a\\
			& \quad F_{a|\vec{r}} \succeq 0, \quad \forall a,\vec{r},\\
			& \quad \tr(F_{a|\vec{r}}) = q_{\vec{r}}\ r_{a}, \quad \forall a,\vec{r},\\
			& \quad \sum_{a} F_{a|\vec{r}} = q_{\vec{r}}\openone, \quad \forall \vec{r}.\\
		\end{aligned}
	\end{equation}
	
	Let us introduce a Lagrangian multiplier for each constraint,
	\begin{equation}\label{SDP_constraints}
		\begin{aligned}
			vE_a +\frac{1-v}{d}\tr(E_a)\openone=\sum_{\vec{r}} F_{a|\vec{r}} \quad &\rightarrow \quad \Gamma_{a}, \quad \forall a\\
			F_{a|\vec{r}} \succeq 0 \quad &\rightarrow \quad Z_{a|\vec{r}}, \quad \forall a,\vec{r},\\
			\tr(F_{a|\vec{r}}) = q_{\vec{r}}\, r_{a} \quad &\rightarrow \quad y_{a|\vec{r}}, \quad \forall a,\vec{r},\\
			\sum_{a} F_{a|\vec{r}} = q_{\vec{r}}\openone \quad &\rightarrow \quad Y_{\vec{r}}, \quad \forall \vec{r}\\
		\end{aligned}
	\end{equation}
	and write down the Lagrangian of the problem:
	\begin{equation}\label{Lagrangian_SDP}
		\begin{aligned}
			\mathcal{L} =& v + \sum_{a} \tr[\Gamma_{a}\left(vE_a +\frac{1-v}{d}\tr(E_a)\openone - \sum_{\vec{r}} F_{a|\vec{r}}\right)]\\
			& + \sum_{a,\vec{r}}\tr(Z_{a|\vec{r}}F_{a|\vec{r}}) + \sum_{a|\vec{r}} y_{a|\vec{r}}[\tr(F_{a|\vec{r}}) - q_{\vec{r}}r_{a}]\\
			& + \sum_{\vec{r}}\tr[Y_{\vec{r}}\left(\sum_{a}F_{a|\vec{r}} - q_{\vec{r}}\openone\right)].
		\end{aligned}
	\end{equation}
	The Lagrangian multipliers will be the variables of the dual problem. Notice that we must restrict  to $Z_{a|\vec{r}} \succeq 0$, since in this case the terms $\tr(Z_{a|\vec{r}}F_{a|\vec{r}})$ are always non-negative for any feasible solution of the primal and therefore the Lagrangian is never smaller than the optimal value of the primal objective function, i.e. $\mathcal{L} \geq v^{*}$.
	
	To make the Lagrangian independent of the primal variables $\lbrace F_{a|\vec{r}},q_{\vec{r}} \rbrace$, we can factorize \eqref{Lagrangian_SDP} as follows:
	\begin{equation}
		\begin{aligned}
			\mathcal{L} =& v\left[1 + \sum_{a}\tr(\Gamma_a E_a) - \frac{1}{d}\sum_{a}\tr(E_a)\tr(\Gamma_a)\right]\\
			& + \sum_{a,\vec{r}}\tr[F_{a|\vec{r}}\left(-\Gamma_a + Z_{a|\vec{r}} + Y_{\vec{r}} + y_{a|\vec{r}}\openone\right)]\\
			& + \sum_{\vec{r}} q_{\vec{r}}\left[-\sum_{a}y_{a|\vec{r}} r_{a} - \tr(Y_{\vec{r}})\right]\\
			& + \frac{1}{d} \sum_{a} \tr(\Gamma_{a})\tr(E_{a}).
		\end{aligned}
	\end{equation}
	Therefore, by imposing the following constraints on the dual variables
	\begin{equation}\label{Dual_constraints}
		\begin{aligned}
			& 1 + \sum_{a}\tr(\Gamma_a E_a) - \frac{1}{d}\sum_{a}\tr(E_a)\tr(\Gamma_a) = 0,\\
			& -\Gamma_a + Z_{a|\vec{r}} + Y_{\vec{r}} + y_{a|\vec{r}}\openone = 0,\\
			& -\sum_{a}y_{a|\vec{r}} r_{a} - \tr(Y_{\vec{r}}) = 0,
		\end{aligned}
	\end{equation}
	the Lagrangian does not present any dependence on the primal variables, i.e. $\mathcal{L} = \frac{1}{d} \sum_{a} \tr(\Gamma_{a})\tr(E_{a})$. Then, the dual problem reads as
	\begin{equation}\label{SDP_dual_visibility}
		\begin{aligned}
			\min_{\substack{\{\Gamma_a\},\{Z_{a|\vec{r}}\},\\ \{y_{a|\vec{r}}\},\{Y_{\vec{r}}\}}} & \quad \frac{1}{d} \sum_{a} \tr(\Gamma_{a})\tr(E_{a})\\
			\text{s.t.}&\quad Z_{a|\vec{r}} \succeq 0, \quad \forall a,\vec{r},\\
			& \quad 1 + \sum_{a}\tr(\Gamma_a E_a) - \frac{1}{d}\sum_{a}\tr(E_a)\tr(\Gamma_a) = 0,\\
			& \quad -\Gamma_a + Z_{a|\vec{r}} + Y_{\vec{r}} + y_{a|\vec{r}}\openone = 0, \quad \forall a,\vec{r},\\
			& \quad -\sum_{a}y_{a|\vec{r}} r_{a} - \tr(Y_{\vec{r}}) = 0. \quad \forall \vec{r}.\\
		\end{aligned}
	\end{equation}
	Note that the $\lbrace Z_{a|\vec{r}} \rbrace$ are slack variables since they do not appear in the objective function. Therefore, we can use the third constraint to solve for them. In the same way, we can rewrite the objective function solving the second constraint. This leads to a simplified version of the dual problem, i.e.
	\begin{equation}\label{SDP_dual_visibility_simplified}
		\begin{aligned}
			\min_{\substack{\{\Gamma_a\}, \{y_{a|\vec{r}}\},\\\{Y_{\vec{r}}\}}} & \quad 1 + \sum_{a}\tr(\Gamma_a E_a)\\
			\text{s.t.} & \quad 1 + \sum_{a}\tr(\Gamma_a E_a) - \frac{1}{d}\sum_{a}\tr(E_a)\tr(\Gamma_a) = 0,\\
			& \quad \Gamma_a - Y_{\vec{r}} - y_{a|\vec{r}}\openone \succeq 0, \quad \forall a,\vec{r},\\
			& \quad -\sum_{a}y_{a|\vec{r}} r_{a} - \tr(Y_{\vec{r}}) = 0. \quad \forall \vec{r}.\\
		\end{aligned}
	\end{equation}
	We can check that Slater's conditions hold for the primal SDP \eqref{SDP_primal_visibility}, i.e. (1) it is possible to find a \textit{strictly-feasible} solution to the problem, (2) the optimal objective is finite. Condition (2) follows from the fact that the first constraint bounds $v$ for every choice of $\mathbf{E}$ except $E_a \propto \openone \forall a$, which can be trivially seen to be simulable. Moreover, choosing a rank-vector $\vec{r^{*}} = (\tr(E_1),\dots,\tr(E_n))$ with $q(\vec{r^{*}}) = 1$ and $F_{a|\vec{r^{*}}} = \frac{\tr(E_a)}{d}\openone$ always gives a solution to the primal in terms of strictly positive $F_{a|\vec{r}}$. Therefore, also condition (1) is satisfied, and strong duality holds \cite{Skrzypczyk_2023}.
	
	The feasible solutions $\lbrace \Gamma_a \rbrace$ of the dual problem \eqref{SDP_dual_visibility_simplified} provide a witness to detect $\mathbf{E} \notin \mathcal{P}$. Indeed, we can define
	\begin{equation}
		W(\mathbf{E}) = \sum_{a} \tr(\Gamma_{a} E_{a}).
	\end{equation}
	From strong duality, we know that $v^{*}(\mathbf{E}) = 1 + \sum_{a} \tr(\Gamma_a E_a)$ for the optimal $\lbrace \Gamma_a \rbrace$. Therefore, for any noisy POVM $\Phi_{v}(\mathbf{E})$ for which $v \leq v^{*}(\mathbf{E})$, we get
	\begin{equation}
		\begin{aligned}
			W(\Phi_{v}(\mathbf{E})) &= v \sum_{a}\tr(\Gamma_a E_a) + \frac{1-v}{d} \sum_{a} \tr(E_{a}) \tr(\Gamma_a)\\
			&= v(v^{*}-1) + (1-v)v^{*} = v^{*} - v \geq 0.
		\end{aligned}
	\end{equation}
	This inequality is satisfied by all $\mathbf{E} \in \mathcal{P}$. Then, finding $W(\mathbf{E}) \ngeq 0$ implies that $\mathbf{E} \notin \mathcal{P}$.

	\section{Witness bounds}\label{AppEfficient}

	In the main text, we introduced a simple witness ansatz to falsify the simulability of a POVM. Consider a rank-1 POVM $\mathbf{\tilde{\mathbf{E}}}$ with elements $\tilde E_a = c_a\ketbra{\psi_a}{\psi_a}$. To show that an arbitrary POVM $\mathbf{E}\notin \mathcal{P}$, we consider a witness inspired by quantum state discrimination,
	\begin{align}\label{eq:witn}
		W_\mathbf{\tilde{E}}(\mathbf{E})=\frac{1}{d}\sum_{a=1}^n  \bracket{\psi_a}{E_a}{\psi_a}.
	\end{align}It can be easily checked that $W_\mathbf{\tilde{E}}(\mathbf{E})\leq 1$ and that the bound is saturated only for $\mathbf E= \mathbf{\tilde{\mathbf{E}}} $. To see if $\mathbf{{\mathbf{E}}} \in \mathcal P$, we obtain an upper bound  on $W_{\mathbf{\tilde{\mathbf{E}}}}$ over all projectively simulable implementations. This can be achieved by applying the SDP constraints \eqref{SDP_primal_visibility} and \eqref{eq:witn} as the objective. We denote the obtained maximum value with $\beta$. To determine $\beta$ in practice, one can leverage the linearity of the problem. For projectively simulable measurements, the witness is a convex combination over deterministic values associated with $\lambda = (\chi,\vec r)$. Due to its linearity, the maximum is achieved for some specific choice of $(\chi,\vec r)$.  We denote the upper bound of the witness for a specific choice of the rank vector as $\beta_{\vec r}$, corresponding to projective simulations considering only the specific rank vector $\vec r$. Each upper bound can be calculated with the witness as the objective using the SDP constraints \eqref{SDP_primal_visibility}. The overall upper bound on the witness is simply $\beta = \max_{\vec r}\beta_{\vec r}$.  More generally, instead of iterating over each rank vector individually, we may package rank vectors into mutually exclusive sets, $\mathcal R_j$. The bound $\beta_{\mathcal R_j}$ corresponding to each set can be calculated with the SDP
	\begin{equation}
		\begin{aligned}
			\max_{\lbrace F_{a|\vec{r}} \rbrace} &\quad \frac{1}{d}\sum_{a=1}^{n} \sum_{\vec{r} \in \mathcal R_j} \bracket{\psi_a}{F_{a|\vec{r}}}{\psi_a}\\
			\text{s.t.}& \quad\tr(F_{a|\vec{r}})=q_{\vec{r}}\ r_a \,\,\forall a, \vec r \in \mathcal R_j,\\
			&\quad \sum_a F_{a|\vec{r}}=q_{\vec{r}}\openone \,\,\forall\, \vec r \in \mathcal R_j, \text{ and } F_{a|\vec{r}}\succeq 0 \,\,\forall a, \vec r \in \mathcal R_j.
		\end{aligned}
	\end{equation}The overall bound is simply $\beta = \text{max}_{j} \beta_{\mathcal R_j}$. Obviously, if each $\mathcal R_j$ contains only one rank vector, we go back to case of going over each rank vector individually.
	
	To simulate SIC measurements, we consider rank vectors of length $N=d^2$, where $d$ is the dimension of the Hilbert space. There is a sharp rise in the number of rank vectors $N_r$ with dimension, with only 10 rank vectors in $d=2$ to more than 4 million rank vectors in $d=6$. Thus, it  quickly  becomes infeasible to solve the SDP individually for each rank vector and obviously to solve it over all rank vectors together. It is therefore useful to package rank vectors in sets of a manageable size, as described above. 
	
	The witness bound $\beta$ can give an upper bound $v_\beta (\mathbf E)$ on the visibility $v^*(\mathbf E)$ by setting $E_a = v_\beta \tilde E_a + \frac{1-v_\beta}{d}\text{tr}(\tilde E_a)\openone$,
	\begin{equation}
		\begin{aligned}
			\beta& = \frac{1}{d}\sum_a \bra{\psi_a} v_\beta\tilde E_a + \frac{1-v_\beta}{d}\text{tr}(\tilde E_a)\openone \ket{\psi_a}\\
			& = v_\beta + \frac{1-v_\beta}{d}.
		\end{aligned}
	\end{equation}

	In Table \ref{tab1}, we give the witness bounds obtained up to $d=6$. 
	
	\begin{center}
		\begin{table}[h!]
			\begin{tabular}{|c| c| c |c|}
				\hline
				Type  & $\beta$ &$v_\beta$\\
				\hline 
				2a$^*$ &$\frac{1}{\sqrt{6}}+\frac{1}{2}$ &  $\sqrt{\frac{2}{3}}$\\
				3a& 0.8821 &0.8232 \\
				3b (Norrell) &$\frac89$ &$\frac{5}{6}$\\
				3c (Hesse)$^*$ &$\frac{4}{9}(1+\cos(\frac{\pi}{9}))$ &$\frac{1}{6}(1+4\cos(\frac{\pi}{9}))$\\
				4a  & 0.8839 & 0.8452 \\
				5a  & 0.8669 & 0.8336 \\
				6a  &0.8655 &0.8386  \\
				\hline
			\end{tabular}\caption{The witness bound $\beta$ and the corresponding visibility $v_\beta$ obtained for SIC-POVMs up to $d=6$. The SIC-POVMs in all dimensions are denoted via the labeling used in \cite{Scott_2010}. The cases for which $v_\beta = v^*$ are denoted with a $*$.}
			\label{tab1}
		\end{table}
	\end{center}

	\section{Projective simulation of noisy Hesse SIC-POVM}\label{AppHesse}
	The decomposition of the noisy Hesse SIC given in the main text splits into 72 projective measurements with 3 effects each. One of the projectors takes the form 
	\begin{align}\label{eq:fidproj}
		\ket{\xi}=\frac{1}{3}\left(\ket{0}+2e^{-5\pi i/9}\ket{1}+2e^{5\pi i/9}\ket{2}\right).
	\end{align}
	The orbit of $\ket{\xi}$ under the group generated by 
	\begin{align}
		X&=\sum_j \ketbra{j\oplus1}{j}, \nonumber \\
		V&=\exp\left(\frac{2\pi i}{3} \ketbra{0}{0}\right), \nonumber \\
		W&=\exp\left(\frac{2\pi i}{3}  \ketbra{\nu}{\nu}\right), \label{eq:v_and_w}
	\end{align}
	where $\ket{\nu}=\frac{1}{\sqrt{3}}\left(\ket{0}+e^{2\pi i/3}\ket{1}+e^{2\pi i/3}\ket{2}\right)$, form the remaining projectors. This results in 216 projectors, which can be divided in 72 sets of three orthogonal states, which form the projective measurements. In the decomposition of the noisy Hesse SIC, each measurement appears with the same probability $q_\lambda = \frac{1}{72}$. $V$ and $W$ actually generate a representation of the group $\mathrm{SL}(2,3) = \langle V,W\rangle$ with 24 elements. Every element in that group leaves the Hesse fiducial vector $\ket{\varphi^{(0)}}$ invariant, and therefore belongs to the stabiliser. The noisy version of the fiducial vector only consists of projectors in the orbit of $\ket{\xi}$ under action of $\mathrm{SL}(2,3)$
	\begin{align}\label{eq:noisy_hesse_fiducial}
		\frac{v}{3}\ketbr{\varphi^{(0)}} + \frac{(1-v)}{9} \openone = \sum_{s\in \mathrm{SL(2,3)}} s \ketbra{\xi}{\xi}s^\dagger .
	\end{align}
	Equation \eqref{eq:noisy_hesse_fiducial} allows one to analytically compute the visibility $v=\frac{1}{6}\left(1+4\cos\left(\frac{\pi}{9}\right)\right)$. The other projectors that make up the remaining POVM elements can be formed by Weyl-Heisenberg operators $\langle X,Z \rangle$. Therefore, each of the 216 projectors is uniquely determined by an element $(s,X^{a_0}Z^{a_1})$, with $s\in \mathrm{SL(2,3)}$.
	
	Furthermore, $V$ and $W$ together with $X$ generate the complete projective Clifford group $\mathrm{PC}(3) \cong \mathrm{SL}(2,3) \ltimes (\mathbb{Z}_3 \times \mathbb{Z}_3)$. The deep connection between SIC-POVMs and the (extended) Clifford group was already found and discussed in Refs.~\cite{Appleby_2005,Scott_2010}. The reason we find this symmetry in the orbit spanned by the projectors might be related to the fact, that the automorphism group of the Hesse SIC is the whole extended Clifford group. For the case of worst-case noise that is introduced in Section~\ref{AppWorstCase}, we observe the same behaviour for the decomposition of the noisy Norrell SIC, where the projectors form an orbit under action of the corresponding automorphism group of the Norrell SIC. Generalising this observation might lead to easier ways of constructing the decomposition of noisy SIC-POVMs and finding their automorphism groups.

	\section{Simulable decompositions and generalisations of flagged SIC-POVMs}\label{AppFlaggs}
	
	\subsection{Flagged SIC2}\label{AppFlaggs2}
	In this section, we construct an analytical decomposition of a noisy version of the flagged SIC2, $\mathbf{E}_{\text{fSIC2}}$, in terms of projective measurements. We generalise the construction for an arbitrary number of dimension $3\leq d \leq 14$ by adding $d-2$ flag effects. The corresponding visibility parameter $v$ yields then a proper lower bound to the simulability threshold. We find that this lower bound matches (up to numerical precision) the upper bound we get from the SDP, establishing that our construction is optimal and the extracted visibilities coincide with the simulability threshold.  
	
	Our model extends the usual two-dimensional SIC-POVM constructed from the fiducial vector $\ket{\varphi_2} = \sqrt{(3 + \sqrt{3})/6}\ket{0} +  e^{-i\pi/4}\sqrt{(3 - \sqrt{3})/6} \ket{1}$ and the corresponding SIC-POVM $\mathbf{E}_\text{SIC2}$ with the four effects $\{E_1^{\text{SIC2}},E_2^{\text{SIC2}},E_3^{\text{SIC2}},E_4^{\text{SIC2}}\}$. In dimension $d$, we define the effects of the flagged SIC2 to consist of the $d+2$ effects
	\begin{align}
		E_i &= E_i^{\text{SIC2}} \oplus 0_{d-2}\quad \text{for }i=1,2,3,4,\\
		E_k &= \ketbr{k-2}\quad \text{for }k=5,\ldots,d+2. \label{eq:E_k}
	\end{align}
	where $0_{l}$ denotes the $l\times l$-dimensional zero matrix. We observe, that coarse graining the last $d-2$ effects \eqref{eq:E_k} to one POVM element does not change the critical visibility computed below (at least in all dimensions we checked numerically). Nonetheless, we look at the more general case above, from which one can always compute the simulation of the coarse grained version.
	
	We start by decomposing the (noisy) POVM in terms of six $d$-dimensional $d$-effect POVMs with unit trace effects. Later, we further decompose these measurements into projective measurements to show that they are projectively simulable.
	
	Inspecting the numerical solution of the SDP suggests an ansatz of six measurements  $\mathbf{F}^{(ij)}$, acting on POVM elements $(i,j,5,6,\ldots,d+2)$ for $1\leq i < j \leq 4$. 
	
	In particular, we choose the ansatz
	\begin{align}
		\begin{pmatrix}
			E_1\\ E_2 \\ E_3 \\ E_4 \\ \hline E_5 \\ \vdots \\ E_{d+2}
		\end{pmatrix} = \frac16\begin{pmatrix}
			F^{(12)}_1 \\ F^{(12)}_2 \\ 0 \\ 0 \\ \hline F^{(12)}_5 \\ \vdots \\ F^{(12)}_{d+2} 
		\end{pmatrix} + \frac16\begin{pmatrix}
			F^{(13)}_1 \\ 0 \\  F^{(13)}_3 \\ 0 \\ \hline F^{(13)}_5 \\ \vdots \\ F^{(13)}_{d+2} 
		\end{pmatrix} 
		+ \ldots + \frac16\begin{pmatrix}
			0 \\ 0 \\F^{(34)}_3 \\ F^{(34)}_4 \\ \hline F^{(34)}_5 \\ \vdots \\ F^{(34)}_{d+2} 
		\end{pmatrix}, 
	\end{align}
	where each of the $\mathbf{F}^{(ij)}$ is a $d$-effect POVM with $\tr(F_k^{(ij)}) = 1$. The horizontal lines are introduced to visually distinguish the effects of the SIC2 from the flags. 
	
	Next, we define the vector
	\begin{align}
		\ket{\theta} = \cos(\theta)\ket{0} + e^{-i\pi/4} \sin(\theta)\ket{1}
	\end{align}
	and choose
	\begin{align}
		F_1^{(12)} &= \lambda \ketbr{\theta} & & + \epsilon (\ketbr{2} + \ldots + \ketbr{d-1}), \nonumber\\
		F_2^{(12)} &= \lambda \ketbr{-\theta} & &+ \epsilon (\ketbr{2} + \ldots + \ketbr{d-1}),\nonumber \\
		F_5^{(12)} &= 2\epsilon \ketbr{1} & &+ \kappa \ketbr{2} + \epsilon (\ketbr{3} + \ketbr{4} + \ldots \nonumber \\
		& & & \pushright{ + \ketbr{d-1}),}\nonumber \\
		F_6^{(12)} &= 2\epsilon \ketbr{1} & &+ \kappa \ketbr{3} + \epsilon (\ketbr{2} + \ketbr{4} + \ldots \nonumber \\
		& & & \pushright{ + \ketbr{d-1}),}\nonumber \\
		\vdots \nonumber\\
		F_{d+2}^{(12)} &= 2\epsilon \ketbr{1} & &+ \kappa \ketbr{d-1} + \epsilon (\ketbr{2} + \ldots \nonumber \\ 
		& & & \pushright{ + \ketbr{d-2}).~~~~~} \label{eq:F12effects}
	\end{align}
	The other five measurements are obtained from this measurement by acting with elements of the automorphism group of SIC2, generated by $U = \exp(2\pi i/3 \ketbr{\varphi_2}) \oplus \openone_{d-2}$ and the Pauli matrices, e.g.,
	$F_1^{(13)} = U F_1^{(12)} U^\dagger$, $F_5^{(13)} =  U F_5^{(12)} U^\dagger$, $F_6^{(13)} =  U F_6^{(12)} U^\dagger$, $\ldots$ 
	and within the measurements according to $F_2^{(12)} = \sigma_Z F_1^{(12)} \sigma_Z$, $F_3^{(13)} = \sigma_X F_1^{(13)} \sigma_X$, 
	$F_4^{(14)} = \sigma_Y F_1^{(14)} \sigma_Y$, etc., where the Pauli matrices are to be understood to act on the first two dimensions only.
	
	Let us now impose that our ansatz constitutes a proper decomposition of the flagged SIC2.
	Imposing $\sum_k F_k^{(ij)} = \openone$ yields a first set of conditions,
	\begin{align}
		2\lambda\cos^2(\theta) = 1,\label{eq:fsiccond1_1}\\
		2(d-2)\epsilon + 2\lambda\sin^2(\theta) = 1,\\
		(d-1)\epsilon + \kappa = 1.
	\end{align}
	Imposing $1/6 \sum_{i<j} F_k^{(ij)}  = vE_k + (1-v)\tr(E_k)\openone/d$ yields for $k=1\ldots 4$ the second set of conditions
	\begin{align}
		\lambda(3 + \cos(2\theta) + 
		\sqrt{2}  \sin(2\theta)) &= (3 + \sqrt{3}) v + \frac6d(1-v),\\
		\lambda(3 - \cos(2\theta) - 
		\sqrt{2}  \sin(2\theta)) &= (3 - \sqrt{3}) v + \frac6d(1-v),\\
		\lambda (\cos(2 \theta) + \sqrt{2} \sin(2 \theta)) &= \sqrt{3} v,\\
		\epsilon &= \frac{1-v}d.
	\end{align}
	Observe that the third line follows from the first two and is therefore linearly dependent.
	
	The same condition evaluated for any $k\geq 5$ yields finally
	\begin{align}
		\epsilon &= \frac{1-v}d,\\
		\kappa &= \frac{(d-1)v + 1}d.
	\end{align}
	However, we have seen both of these conditions before (the second one follows from the third line of the first set of conditions). The other equations fix everything and we obtain as a solution
	\begin{align}
		\epsilon = \frac{1}{4+d+(d-2)\sqrt{3}+2\sqrt{6+2\sqrt{3}(d-2)}}.
	\end{align}
	
	From this, we get $v$ via  
	\begin{align}
		v &= 1-d\epsilon \nonumber \\
		&= \frac{\sqrt{3} d +2 \sqrt{2 \sqrt{3} d -4 \sqrt{3}+6}-2 \sqrt{3}+4}{d + \sqrt{3} d  +2 \sqrt{2 \sqrt{3} d -4 \sqrt{3}+6}-2 \sqrt{3}+4}. \label{eq:visibilityfSIC2}
	\end{align}
	We can also obtain the remaining parameters easily. For instance, using the first line of the first set of conditions to eliminate $\lambda$ in the second line yields a formula for $\tan^2(\theta)$. With this, we can easily confirm that $\tr(F^{(ij)}_k) = 1$.
	
	Note that the expression for $v$ recovers the correct value of $\sqrt{2/3}$ for  $d=2$. For $d=3$, we know that each $d$-dimensional $d$-effect POVM with unit trace effects is projectively simulable \cite{Oszmaniec2017} and we get the visibility of $v \approx 0.798$, which implies (together with the corresponding upper bound from the SDP that confirms that our construction is optimal) that it is more robust than the Norrell SIC but less robust than the Hesse SIC. 
	
	For $d\geq 4$, we still have to show that the $\mathbf{F}^{(ij)}$ are simulable. To that end, we explicitly derive a projective simulation for $\mathbf{F}^{(12)}$ for dimensions $3\leq d \leq 14$. Due to the unitary symmetries, this model can then be directly transformed into simulations of the other measurements.
	
	Consider the measurement $\mathbf{F}^{(12)}$ with its $d$ effects given in Eq.~\eqref{eq:F12effects}. We decompose this POVM in terms of $2^{d-2}$ projective measurements $\mathbf{P}^{(s_2,\ldots,s_{d-1})}$ where $s_i \in \{-1,1\}$. For each fixed choice of signs $s_2,\ldots,s_{d-1}$, we choose the effects of the corresponding projective measurement as
	\begin{align}
		P^{(s_2,\ldots,s_{d-1})}_1 &= \ketbr{\pi_1}\nonumber \\
		\text{ with }\ket{\pi_1} &= \sqrt{\lambda}\ket{\theta} + \sqrt{\epsilon}(s_2\ket{2} + \ldots +s_{d-1}\ket{d-1}),\nonumber \\
		P^{(s_2,\ldots,s_{d-1})}_2 &= \ketbr{\pi_2}\nonumber \\
		\text{ with }\ket{\pi_2} &= \sqrt{\lambda}\ket{-\theta} - \sqrt{\epsilon}(s_2\ket{2} + \ldots+ s_{d-1}\ket{d-1}),\nonumber \\
		P^{(s_2,\ldots,s_{d-1})}_5 &= \ketbr{\pi_5}\nonumber \\ 
		\text{ with }\ket{\pi_5} &= \sqrt{2\epsilon}\ket{1} + \sqrt{\kappa}e^{i\alpha}s_2\ket{2}  \nonumber \\
		& \phantom{=}+\sqrt{\epsilon}e^{i\beta}(s_3\ket{3} + \ldots+ s_{d-1}\ket{d-1}),\nonumber \\
		\vdots&\nonumber \\
		P^{(s_2,\ldots,s_{d-1})}_{d+2} &= \ketbr{\pi_{d+2}}\nonumber \\
		\text{ with }\ket{\pi_{d+2}} &= \sqrt{2\epsilon}\ket{1} + \sqrt{\kappa}e^{i\alpha}s_{d-1}\ket{d-1} \nonumber \\
		&\phantom{.} + \sqrt{\epsilon}e^{i\beta}(s_2\ket{2} + \ldots+ s_{d-2}\ket{d-2}).
	\end{align}
	Considering the convex combination $\frac1{2^{d-2}} \sum_{s_2,\ldots,s_{d-1}} \mathbf{P}^{(s_2,\ldots,s_{d-1})}$ yields then the corresponding effects of $\mathbf{F}^{(12)}$. Note that we consider all possible signs to have the off-diagonal terms cancel each other.
	
	Finally, we have to ensure that the $\ket{\pi_i}$ are orthogonal. These conditions are independent of the choice of signs and reduce to the three independent conditions
	\begin{align}
		0=\braket{\pi_1}{\pi_2} &= \lambda \cos(2\theta) - (d-2)\epsilon,\\
		0=\braket{\pi_1}{\pi_5} &= \sqrt{\xi - d +3 }e^{i\pi/4} + \sqrt{\xi}e^{i\alpha} + (d-3)e^{i\beta},\\
		0=\braket{\pi_5}{\pi_6}&=2\sqrt{\xi}\cos(\alpha - \beta) + d - 2,
	\end{align}
	where $\xi = \frac{1}{\epsilon} - d + 1 = 5 + \sqrt{3}(d-2) + 2\sqrt{6+2\sqrt{3}(d-2)}$.
	The first of these conditions follows automatically from the trace condition $\lambda +(d-2)\epsilon = 1$ together with Eq.~\eqref{eq:fsiccond1_1}. The second condition can be fulfilled for suitable choices of $\alpha$ and $\beta$ as long as the sum of any two of the three magnitudes of the complex numbers exceeds the third. This, however, only works up to dimensions $d=14$, and one can check that for each dimension $d\in \{3,\ldots,14\}$, the solutions to the second condition also fulfil the third. This shows that our model yields a proper decomposition for dimensions $3\leq d \leq 14$, establishing a lower bound on the visibility threshold. Using the SDP method, we find matching upper bounds for these cases confirming that the simulability threshold of $\mathbf{E}_{\text{fSIC2}}$ is given by Eq.~\eqref{eq:visibilityfSIC2} at least up to dimension $d\leq 14$. In particular, for $d=4$, this threshold value of $v^*\approx 0.7856$ is below the simulability threshold of the four-dimensional SIC-POVM of $v^*(\mathbf{E}_{\text{SIC4}}) \approx 0.826$.
	
	In dimensions $d\leq 16$ where we ran the SDP, we know that \eqref{eq:visibilityfSIC2} is an upper bound on $v_d$, i.e. the visibility corresponding to the most noise-robust POVM. We conjecture this to hold true in all dimensions. In the limit of large dimensions (even though we have no projective model in these cases), our value of $v$ approaches $v
	\longrightarrow \frac{1}{2} \left(3-\sqrt{3}\right) \approx 0.6340$.

	\subsection{Flagged Hesse SIC}\label{AppFlaggs3}
	
	Here, we construct a simulation of the flagged Hesse SIC. We follow the same approach as in the case of the flagged SIC2, i.e., we fix the dimension $d$ and add $d-3$ flag effects, such that we obtain the $d+6$ effects
	\begin{align}
		E_i &= E_i^{\text{Hesse}} \oplus 0_{d-3}\quad \text{for }i=1,\ldots9,\\
		E_k &= \ketbr{k-7}\quad \text{for }k=10,\ldots,d+6. \label{eq:E_k_fsic3}
	\end{align}

	We start again by decomposing the noisy version of this POVM in terms of $d$-dimensional $d$-effect POVMs with unit trace effects before we further decompose these measurements into projective measurements.
	
	Inspecting the numerical solution of the SDP suggests an ansatz of 72 measurements  $\mathbf{F}^{(ijk)}$ with non-zero effects $F_l^{(ijk)}$ for $l\in \{i,j,k,10,11,\ldots,d+6\}$. Here, $1\leq i<j<k\leq 9$ are chosen such that the corresponding measurement is present in the decomposition of the unflagged Hesse SIC as explained in Section~\ref{AppHesse}. Let us denote the set of size 72 of $(i,j,k)$ corresponding to a measurement by $\mathcal{S}$, such that $E_a = \frac1{72} \sum_{(i,j,k)\in\mathcal{S}}F_a^{(ijk)}$.  Thus, each of the $\mathbf{F}^{(ijk)}$ is selected with the same probability of $1/72$.
	
	Next, we define the vector
	\begin{align}
		\ket{\theta,\phi} = \cos(\theta)\ket{0} + \sin(\theta) e^{i\phi}/\sqrt{2} (\ket{1} + \ket{2}).
	\end{align}
	The fiducial projector $\ket{\xi}$ in Eq.~\eqref{eq:fidproj} was selected from measurement $i,j,k = 1,6,9$. Thus, we start with an ansatz or $F^{(169)}$ with the effects 
	\begin{align}
		F_1^{(169)} &= \lambda \ketbr{\theta,\phi}  + \epsilon (\ketbr{3} + \ldots + \ketbr{d-1}), \nonumber\\
		F_6^{(169)} &= \lambda h_6V\ketbr{\theta,\phi}V^\dagger h_6^\dagger + \epsilon (\ketbr{3} + \ldots \nonumber \\
		&  \pushright{+ \ketbr{d-1})},\nonumber \\
		F_9^{(169)} &= \lambda h_9V^2\ketbr{\theta,\phi}(V^\dagger)^2 h_9^\dagger + \epsilon (\ketbr{3} + \ldots \nonumber \\
		&  \pushright{+ \ketbr{d-1})},\nonumber \\
		F_{10}^{(169)} &= 3\epsilon \ketbr{\alpha} + \kappa \ketbr{3} + \epsilon (\ketbr{4} + \ketbr{5} + \ldots \nonumber \\
		&  \pushright{+ \ketbr{d-1})},\nonumber \\
		\vdots \nonumber\\
		F_{d+6}^{(169)} &= 3\epsilon \ketbr{\alpha} + \kappa \ketbr{d-1} + \epsilon (\ketbr{3} + \ldots \nonumber \\ 
		&  \pushright{+ \ketbr{d-2}),~~~~~} \label{eq:F169effects}
	\end{align}
	and all other effects being zero.
	Here, $h_6 = XZ^2$ and $h_9=X^2Z^2$ denote the corresponding Weyl-Heisenberg matrices, $V = \exp\left(\frac{2\pi i}{3} \ketbra{0}{0}\right)$ is one of the generators of the symmetry group of the decomposition of the Hesse SIC defined in Eq.~\eqref{eq:v_and_w} and $\ket{\alpha} = \frac1{\sqrt{3}}(\ket{0} + e^{-8\pi i/9}\ket{1} + e^{8\pi i/9}\ket{2}).$
	The other 71 measurements are obtained from this measurement by acting with the rest of the symmetry group (acting trivially on the flag dimensions). 
	
	Imposing $\sum_l F_l^{(169)} = \openone$ yields the conditions
	\begin{align}
		1 &= (d-3)\epsilon + \lambda,\label{eq:fsic3cond1_1}\\
		0 &= 2(d-3)\epsilon e^{8\pi i/9} + 2\sqrt{2}\lambda c_\theta s_\theta e^{-i\phi} + \lambda s_\theta^2 e^{2\pi i/3+2i\phi},\label{eq:fsic3cond1_2}\\
		1 &= (d-1)\epsilon + \kappa. \label{eq:fsic3cond1_3}
	\end{align}
	Here, we introduced the abbreviations $c_\theta = \cos(\theta)$, $s_\theta = \sin(\theta)$.
	
	Imposing $1/72 \sum_{(i,j,k)\in \mathcal{S}} F_l^{(ijk)}  = vE_l + (1-v)\tr(E_l)\openone/d$ yields for $l=1\ldots 9$ the conditions
	\begin{align}
		v &= \lambda \frac{3s_\theta^2 s_\phi^2 - 1}{2},\label{eq:fsic3cond2_1}\\
		v &= 1-d\epsilon. \label{eq:fsic3cond2_2}
	\end{align}
	
	Finally, we evaluate the sum for $l\geq 10$. Here, the vector $\ket{\alpha}$ is chosen such that its orbit under action of the 216 elements of the projective Clifford group $\text{Cl}$ generated from $V$, $W$ and the shift matrix $X$ given in Eq.~\eqref{eq:v_and_w}, contains only 72 different vectors (one for each measurement). However, as the Clifford group forms a unitary three-design \cite{Webb_2016}, it follows that 
	\begin{align}
		\frac{1}{216} \sum_{c\in \text{Cl}} c \ketbr{\alpha} c^\dagger = \frac13 (\ketbr{0} + \ketbr{1} + \ketbr{2}).  
	\end{align}
	Thus, choosing for instance $l=10$, we obtain 
	\begin{align}
		\frac1{72} \sum_{(i,j,k) \in \mathcal{S}} F_{10}^{(ijk)} = \epsilon(\ketbr{0} + \ketbr{1} + \ketbr{2})  \nonumber \\
		\pushright{+ \kappa\ketbr{3} + \epsilon (\ketbr{4} + \ldots + \ketbr{d-1}).}
	\end{align}
	Setting this equal to $v \ketbr{4} + \frac{1-v}{d}\openone$ yields no new conditions.
	
	Combining Eqs.~\eqref{eq:fsic3cond1_1}, \eqref{eq:fsic3cond2_1}   and \eqref{eq:fsic3cond2_2}, we obtain an expression for $v$ which solely depends on the angles $\theta$ and $\phi$:
	\begin{align}
		v = \frac{3}{\frac{2d}{3s_\theta^2 s_\phi^2 - 1} - (d-3)}. 
	\end{align}
	From this, we can also express $\epsilon = (1-v)/d$, $\lambda=1-(d-3)\epsilon$ and $\kappa = 1-(d-1)\epsilon$ in terms of $\theta$ and $\phi$. Inserting this into the condition of Eq.~\eqref{eq:fsic3cond1_2} yields a consistency equation for $\theta$ and $\phi$, namely
	\begin{align}
		(d-3)  \left(1 - s_\theta^2 s_\phi^2\right)e^{8\pi i/9 + i \phi} +   s_\theta^2e^{2\pi i/3 + 3 i \phi} + 2 \sqrt{2} c_\theta s_\theta = 0.
	\end{align}
	While we are not able to find closed form solutions, this consistency equation seems to have finitely many solutions in dimension $d\geq 4$, which we can obtain easily numerically up to arbitrary precision. We select the solution which yields the highest visibility and obtain from this numbers matching perfectly the upper bound from the SDP (we checked this up to dimension $d=9$). In particular, we obtain in $d=4$, $v_4 \approx 0.78233002$, in $d=5$, $v_5 \approx 0.77339360$ and in $d=6$, $v_6 \approx 0.76576302$.
	
	Due to the similarity of the decomposition of the flagged Hesse SIC to that of the flagged SIC2, we can choose a completely analogous ansatz for a further decomposition of each of the $\mathbf{F}^{(ijk)}$ in terms of $2^{d-3}$ projective measurements and we checked numerically that this yields a complete projective simulation of the flagged Hesse SIC matching the SDP upper bound at least in dimensions $4\leq d \leq 9$ (where in $d=9$, we otain $v_9 \approx 0.74762609$). Thus, the reported values of $v$ for the flagged Hesse SIC in these dimensions are known to be tight and correspond to the projective simulability thresholds.
	
	\section{Worst case noise}\label{AppWorstCase}
	In this section, we further discuss the problem of quantifying the non-projectivity of a POVM $\mathbf{E}$ subject to the worst-case noise model.
	Analogous the the depolarising noise, we define the noise map
	\begin{align}
		\tilde\Phi_{v}(X,X_\text{noise}) = vX + (1-v)X_\text{noise}
	\end{align}
	and the worst-case visibility threshold as
	\begin{equation}
		\tilde v^*(\mathbf{E})=\max_{\mathbf{E}_\text{noise}}  \{ v| \{\tilde\Phi_v(E_a, E^a_{\text{noise}})\}_a\in \mathcal{P}\},
	\end{equation}
	where we additionally maximize over all POVMs $\mathbf{E}_\text{noise} = \{E^a_{\text{noise}}\}_{a=1}^n$.
	Choosing $E^a_{\text{noise}} = \tr(E_a)\openone /d$ shows that $v^*(\mathbf{E})\leq \tilde v^*(\mathbf{E})$.
	
	What is more, in contrast to the usual visibility threshold, the worst-case visibility threshold cannot decrease when adding trivial flags to it:
	
	\begin{lemma} \label{lem:flagging_worstcase}
		Let $\mathbf{E} = \{E_a\}_{a=1}^n$ be a $d$-dimensional POVM and $\mathbf{E}_f = \{E^a_f\}_{a=1}^{n+1}$ the $(d+1)$-dimensional flagged POVM with effects
		\begin{align}
			E^a_f = \begin{cases} E_a \oplus 0 & \text{for }a=1,\ldots,n, \\ 0_d \oplus 1 & \text{for }a = n+1\end{cases},
		\end{align}
		where $0_d$ denotes the $d\times d$-dimensional zero matrix.
		Then $\tilde v^*(\mathbf{E}_f) \geq \tilde v^*(\mathbf{E})$.
	\end{lemma}
	\begin{proof}
		Let $\tilde v^*$ be the worst-case visibility threshold of $\mathbf{E}$ and $\mathbf{E}^*_\text{noise}=\{F_a\}_{a=1}^n$ be the corresponding worst-case noise, such that $\tilde v^* \mathbf{E} + (1-\tilde v^*)\mathbf{E}^*_{\text{noise}} \in \mathcal{P}$. We define the $(d+1)$-dimensional POVM $\mathbf{E}^*_{\text{noise},f}$ via its $n+1$ effects
		\begin{align}
			F_f^a = \begin{cases} F_a \oplus 0 & \text{for }a=1,\ldots,n, \\ 0_d \oplus 1 & \text{for }a = n+1\end{cases}.
		\end{align}
		Note that the last effect of POVMs $\mathbf{E}_f$ and $\mathbf{E}^*_{\text{noise},f}$ match, such that $\tilde v^* E^{n+1}_f + (1-\tilde v^*)F^{n+1}_f = 0_d \oplus 1$. Thus, we can take a projective decomposition of $\tilde v^* \mathbf{E} + (1-\tilde v^*)\mathbf{E}^*_{\text{noise}}$ and append to each projective measurement the trivial projector $\ketbr{d}$, yielding a projective decomposition of $\tilde v^* \mathbf{E}_f + (1-\tilde v^*)\mathbf{E}^*_{\text{noise},f}$. Thus, $\tilde v^*$ is contained in the set that is maximized over to obtain $\tilde v^*(\mathbf{E}_f)$, which yields the claim.
	\end{proof}
	
	We can modify our SDP to calculate upper bounds on the worst-case visibility threshold instead. It reads as
	\begin{equation}\label{SDP_primal_visibility_worst_case}
		\begin{aligned}
			\max_{\substack{v,\{F_{a|\vec{r}}\},\\\{E^{a}_{\text{noise}}\}}} & \quad v\\
			\text{s.t.}&\quad v E_a+ (1-v) E^{a}_{\text{noise}} =  \sum_{\vec{r}} F_{a|\vec{r}}, \quad \forall a,\\
			& \quad F_{a|\vec{r}} \succeq 0, \quad \forall a,\vec{r},\\
			& \quad \tr(F_{a|\vec{r}}) = q_{\vec{r}}\ r_{a}, \quad \forall a,\vec{r},\\
			& \quad \sum_{a} F_{a|\vec{r}} = q_{\vec{r}}\openone, \quad \forall \vec{r},\\
			& \quad E^{a}_{\text{noise}} \succeq 0, \quad \forall a,\\
			& \quad \sum_{a}E^{a}_{\text{noise}} = (1-v)\openone,\\
		\end{aligned}
	\end{equation}
	where $\mathbf{E}_\text{noise} = \lbrace E^{a}_{\text{noise}} \rbrace_{a}$ is an arbitrary POVM that denotes the ``worst'' noise.
	
	Following the same approach presented in Section \ref{AppDual}, we find the following dual problem,
	\begin{equation}\label{SDP_dual_visibility_simplified_worst_case}
		\begin{aligned}
			\min_{\substack{\{\Gamma_a\}, \{y_{a|\vec{r}}\},\\ \{Y_{r}\},\beta}} & \quad 1 + \sum_{a}\tr(\Gamma_a E_a)\\
			\text{s.t.} & \quad 1 + \sum_{a}\tr(\Gamma_a E_a) +\tr(\beta) = 0,\\
			& \quad \Gamma_a - Y_{r} - y_{a|\vec{r}}\openone \succeq 0, \quad \forall a,\vec{r},\\
			& \quad -\sum_{a}y_{a|\vec{r}} r_{a} - \tr(Y_{\vec{r}}) = 0, \quad \forall \vec{r},\\
			& \quad -\Gamma_a - \beta \succeq 0, \quad \forall a.\\
		\end{aligned}
	\end{equation}
	that can be used in the algorithm for the search of the most non-projective POVM, as explained in the main text. Our implementation is available at \cite{github-code}.
	
	The results of the evaluations for different POVMs can be found in Section \ref{AppTables}. Notably, for the Hesse SIC-POVM, the worst-case noise turns out to be the POVM living in its orthogonal complement, i.e. $E^{a}_{\text{noise}} = \frac{1}{6} \openone - \frac{1}{2}\ketbra{\psi_{a}}$. Also for the Norrell SIC-POVM and for the SIC in $d = 2$, the worst-case noise turns out to be orthogonal to the POVM itself: as a consequence, the threshold visibility turns out to be equal to the witness bound $\beta$ in Table \ref{tab1}.
	Moreover, we note that, unlike for depolarisation noise, in the worst-case noise, both $\mathbf{E}_\text{fSIC2}$ and $\mathbf{E}_\text{fSIC3}$ show a simulability threshold higher than their unflagged version in accordance with Lemma~\ref{lem:flagging_worstcase}.

	\section{Result tables}\label{AppTables}
	\begin{center}
		\begin{table}[h!]
			\begin{tabular}{|c|c|c|c|}
				\hline
				Dim. & Type & $v_{\text{depol}}$ & $v_{\text{worst}}$ \\
				\hline
				\hline
				2 & $\text{most}_{d=2}$ & 0.8165*	& 0.9082*\\
				2 & 2a & 0.8165* & 0.9082*\\
				\hline
				3 & $\text{most}_{d=3}$ & 0.7931* & 0.8621*\\
				3 & 3a & 0.8003* & 0.8687*\\
				3 & 3b (Norrell) & 0.8058* & $\frac{8}{9}$*\\
				3 & 3c (Hesse) & 0.7931* & 0.8621*\\ 
				3 & fSIC2 & $0.7985^\dagger$ & 0.9137\\ 
				\hline
				4 & $\text{most}_{d=4}$ & $0.7823^\dagger$ & 0.8681\\
				4 & 4a & 0.8255* & 0.8739*\\
				4 & fSIC3 & $0.7823^\dagger$ & 0.8704\\
				\hline
			\end{tabular}
			\caption{Threshold visibility evaluated through the SDP for depolarising and ``worst-case'' noise model in \eqref{SDP_primal_visibility} and \eqref{SDP_primal_visibility_worst_case}. The SIC-POVMs in dimension $d=2,3,4$ are denoted via the labelling used in \cite{Scott_2010}. fSIC2 and fSIC3 denote the flagged SIC-POVMs introduced in the main text. The notation $\text{most}_{d=x}$ indicates the threshold visibility for the most non-projective POVM in dimension $d$ found using our numerical search algorithm. The * denotes the values for which the SDP gives exactly a projective simulation. The $^\dagger$ denotes the values, where the SDP does not give projective simulations, but measurements for which we found projective simulations by hand (see Section \ref{AppFlaggs}).}
		\end{table}
	\end{center}

\end{document}